\documentclass[article]{IEEEtran}

\usepackage{research5IEEE}

\usepackage[utf8]{inputenc} \ifCLASSINFOpdf
  
\else
  
\fi

\newcommand{\BigPrvcond}[2]{\Pr\Bigl[#1 \kern-0.1em \Bigm| \kern-0.1em#2\Bigr]}

\newtheorem{assumption}{Assumption}
\newtheorem{theorem}{Theorem}
\newtheorem{lemma}[theorem]{Lemma}
\newtheorem{corollary}[theorem]{Corollary}

\newtheorem{proposition}[theorem]{Proposition}
\newtheorem{remark}[theorem]{Remark}

\newcommand{\logplus}{\operatorname{log}^{+}}
\newcommand{\RealsP}{\Reals_{+}}
\newcommand{\RealsPP}{\Reals_{++}}

\usepackage{definitions}

\begin{document}

\title{Constrained Source Coding with Side Information}

\author{Amos Lapidoth, Andreas Mal\"ar, and Mich\`ele
  Wigger\thanks{The material in this paper was presented in part at
    the 2011 Information Theory and Applications Workshop and at the
    2011 IEEE International Symposium on Information Theory.

A.~Lapidoth is with the Department of
    Information Technology and Electrical Engineering, ETH Zurich,
    Switzerland. (email: lapidoth@isi.ee.ethz.ch). A.~Mal\"ar was with
    the Department of Information Technology and Electrical
    Engineering, ETH Zurich, Switzerland.  He is now with Malcom AG, Zurich, Switzerland 
 (email: andreas@malcom.ch). 
  M.~Wigger is with the Communications and Electronics Department, 
Telecom ParisTech, Paris, France (email: michele.wigger@telecom-paristech.fr).

The work of A.~Mal\"ar  was supported by
    an IDEA League student grant. The work of M.~Wigger was supported by the "Emergences" grant of the
    city of Paris.}
}
\maketitle

\allowdisplaybreaks[4]

\begin{abstract}
  The source-coding problem with side information at the decoder is
  studied subject to a constraint that the encoder---to whom the side
  information is unavailable---be able to compute the decoder's
  reconstruction sequence to within some distortion.

  For discrete memoryless sources and finite single-letter distortion
  measures, an expression is given for the minimal description rate as
  a function of the joint law of the source and side information and
  of the allowed distortions at the encoder and at the decoder. The minimal
  description rate is also computed for a memoryless Gaussian source
  with squared-error distortion measures.
  
  A solution is also provided to a more general problem where there
  are more than two distortion constraints and each distortion
  function may be a function of three arguments: the source symbol,
  the encoder's reconstruction symbol, and the decoder's reconstruction
  symbol.
\end{abstract}

\IEEEpeerreviewmaketitle

\section{Introduction}

\IEEEPARstart{L}{ike} Wyner and Ziv \cite{wynerziv76}, we study a
setting where a sequence generated by a source is to be described
succinctly to a reconstructor (``decoder'') with access to some side
information. Wyner and Ziv showed that, although the side information
is not available at the describing terminal (``encoder''), it can be
beneficial in improving the trade-off between the rate of description
and the reconstruction distortion.  They fully characterized this
trade-off for memoryless sources with single-letter distortion
measures. Unlike the case without side information---since the side
information is used in the reconstruction process, and since the side
information is not available at the describing terminal---the
describing terminal cannot tell how the source sequence it observes
will be reconstructed. In some settings, this is
unacceptable. Steinberg \cite{steinberg09} therefore studied the
common-reconstruction problem where an additional restriction is
imposed that the reconstruction sequence be computable with
probability nearly one at the describing terminal. This greatly limits
the extent by which the reconstruction can depend on the side
information. More generally, there is a tension between the degree by
which the reconstructing terminal utilizes the side information and the
precision with which the describing terminal can compute the
reconstruction sequence.  It is this tension that we study in this
paper.

\begin{figure}[tbp]
\centering

\setlength{\unitlength}{0.8cm}
\begin{picture}(7.7,3.8)(-4.5,-2.4)

\put(-5,-0.75){\vector(1,0){1}}
\put(-4.8,-0.5){$X^n$}

\put(2.5,1){\vector(0,-1){1}}
\put(2.7,0.5){$Y^n$}

\put(-1.5,-0.5){$M$}
\put(-2,-0.75){\vector(1,0){3.5}}

\put(2.5,-1.5){\vector(0,-1){1}}
\put(2.75,-2.25){$\hat{X}_{\d}^n$}

\put(-3,-1.5){\vector(0,-1){1}}
\put(-2.8,-2.25){$\hat{X}_{\e}^n$}

\put(-0.3,-3.5){$\displaystyle\frac{1}{n}\sum_{i=1}^{n}\Exp\big[d_\d(X_i,\hat{X}_{\d,i})\big]
\leq \Dd$}
\put(-6,-3.5){$\displaystyle\frac{1}{n}\sum_{i=1}^{n}\Exp\big[d_{\e}(\hat{X}_{\d,i},\hat{X}_{\e,i})\big]
\leq \De;$}

\put(-3.75,-0.85){encoder}
\put(1.8,-0.85){decoder}

\put(-4.0,-1.5){\line(0,1){1.5}}
\put(-4.0,0){\line(1,0){2}}
\put(-2,-1.5){\line(0,1){1.5}}
\put(-2,-1.5){\line(-1,0){2}}
\put(1.5,-1.5){\line(0,1){1.5}}
\put(1.5,0){\line(1,0){2}}
\put(3.5,-1.5){\line(0,1){1.5}}
\put(3.5,-1.5){\line(-1,0){2}}

\end{picture}
\vspace{10.5mm}
  \caption{Constrained Wyner-Ziv coding.}
  \label{fig:problem-diagram}
\end{figure}

To quantify this tension, we require that the describing terminal
generate an estimate of the sequence that will be produced at the
reconstructing terminal (Figure~\ref{fig:problem-diagram}). We then
study the distortions that can be simultaneously achieved at the
describing terminal (''the encoder distortion'') and at the
reconstructing terminal (''the decoder distortion'') as a function of
the description rate.  If the encoder's distortion function is the
Hamming distance and if the allowed distortion is zero, then our
problem reduces in essence to Steinberg's common-reconstruction
problem.\footnote{For a precise statement see
  Remark~\ref{rem:steinberg} in Section~\ref{sec:related} ahead.} And
if the allowed encoder distortion is infinite, our problem reduces to
Wyner and Ziv's problem.  We can thus view our problem as a
generalization of the Wyner-Ziv problem and Steinberg's common
reconstruction problem.

For discrete memoryless sources and finite single-letter distortion
functions, we provide a single-letter characterization of the
trade-off between the description rate and the distortions at the
encoder and decoder sides. We also calculate this trade-off for a
memoryless Gaussian source and squared-error distortion
functions. Finally,
%
%
%
in Section~\ref{sec:Kdistortions}, we generalize the results to account
for more than two constraints and to allow each distortion function to
depend on three arguments: the source symbol, the encoder's
reconstruction symbol, and the decoder's reconstruction symbol.

Steinberg's work was also extended in other ways.
Kittichokechai, Oechtering, and Skoglund \cite{Kittipongetal}
determined the rate-distortion function under a common-reconstruction
constraint for a modified Wyner-Ziv setup where the encoder can
influence the decoder's side information via an
action-generator. Timo, Grant, and Kramer \cite{TimoGrantKramer2010},
\cite{TimoGrantKramer2012} and Ahmadi, Tandon, Simeone, and Poor
\cite{AhmadiTandonSimeonePoor2012}, \cite{AhmadiTandonSimeonePoor2013}
derived the rate-distortions function under a common-reconstruction
constraint for two special cases of the Heegard-Berger/Kaspi problem
(the Wyner-Ziv problem with two decoders):
\cite{AhmadiTandonSimeonePoor2012}, \cite{AhmadiTandonSimeonePoor2013}
for physically degraded side informations, and
\cite{TimoGrantKramer2010}, \cite{TimoGrantKramer2012} for
complementary
side informations. 
Ahmadi, Tandon, Simeone, and Poor \cite{AhmadiTandonSimeonePoor2012},
\cite{AhmadiTandonSimeonePoor2013} also presented the
rates-distortions function under a common-reconstruction constraint
for a cascade source-coding problem when the side informations are
physically degraded. Finally, already in~\cite{steinberg09}, Steinberg
studied the implications of the common-reconstruction constraint on
the simultaneous transmission of data and state and on joint
source-channel coding for the degraded broadcast channel.

The paper is organized as follows. In the rest of this section we
introduce our notation. In Section~\ref{sec:setup} we treat discrete
sources and general distortions, and in Section~\ref{sec:Gauss}
Gaussian sources with quadratic distortions. In
Section~\ref{sec:Kdistortions} we revisit discrete sources but this
time with more and more general distortion constraints.

\subsection{Notation}

Random variables are denoted by upper-case letters and their
realizations by lower-case letters. Vectors are denoted by bold-face
letters: random vectors by upper-case bold-face letters, and
deterministic vectors by lower-case bold-face letters. Sets and events
are denoted by calligraphic letters, i.e., $\set{A}$. An $n$-tuple
$(A_1,\ldots,A_n)$ is denoted $A^{n}$, and the $n$-fold Cartesian
product of the set $\set{A}$ is denoted $\set{A}^n$. The convex hull of a set $\set{A}$ is denoted by $\textnormal{conv}(\set{A})$. To indicate that
the random variables $A$ and $C$ and conditionally independent given
$B$ we write
\begin{equation*}
  A \markov B \markov C.
\end{equation*}

The transpose of a vector~$\vect{a}$ is denoted by $\trans{\vect{a}}$;
its Euclidean norm by $\|\vect{a}\|$; and the Euclidean inner product
between the vectors $\vect{a}$ and $\vect{b}$ by
$\inner{\vect{a}}{\vect{b}}$. The set of real numbers is denoted
$\Reals$ and its $d$-fold Cartesian product $\Reals^{d}$. The
nonnegative reals are denoted $\RealsP$, and the positive reals
$\RealsPP$.  The respective $d$-fold Cartesean products are denoted
$\RealsP^{d}$ and $\RealsPP^{d}$.
We use $\mathrm{I}(\cdot)$ to denote the indicator function:
$\mathrm{I}(\textnormal{statement})$ is equal to one if the statement
is true and is equal to zero if it is false. Throughout the paper
$\log(\cdot)$ denotes base-2 logarithm, and $\logplus(\xi) =
\max\{\log \xi, 0\}$. The abbreviation IID stands for independently
and identically distributed.

\section{Discrete Memoryless Source and General Distortions}\label{sec:setup}
\subsection{Problem Statement}\label{subsec:setup}
Our setting is illustrated in~Figure~\ref{fig:problem-diagram} and is
specified by a tuple
\begin{equation*}
  \bigl(\mathcal{X}, \mathcal{Y}, \hat{\mathcal{X}}, P_{XY},
  \dd, \de, \Dd, \De \bigr),
\end{equation*}
where $\mathcal{X}, \mathcal{Y}, \hat{\mathcal{X}}$ are finite sets,
$P_{XY}$ is a probability distribution on $\mathcal{X}\times
\mathcal{Y}$; $\dd(\cdot,\cdot)$ and $\de(\cdot,\cdot)$ are nonnegative functions 
\begin{align}
\dd & \colon \mathcal{X} \times \hat{\mathcal{X}} \to \distcodomain \label{eq:ahava_d}\\
\de & \colon \hat{\mathcal{X}} \times \hat{\mathcal{X}} \to
\distcodomain; \label{eq:ahava_e}
\end{align}
and $\Dd$ and $\De$ are nonnegative real numbers.

The sets $\mathcal{X}$,  $\mathcal{Y}$, and $\hat{\mathcal{X}}$ model the source, side information, and reconstruction alphabets. A source sequence $X^n\in\set{X}^n$ is observed
 at the encoder (but not at the decoder) and a side-information sequence $Y^n\in\set{Y}^n$
 at the decoder (but not at the encoder). 
 The sequence of pairs $\{(X_i,Y_i)\}_{i=1}^{n}$ is assumed to be drawn IID according to the joint law $P_{XY}$. 
 
The encoder describes the source sequence $X^{n}$ to the decoder by an index
\begin{IEEEeqnarray}{rCl}\label{eq:enc}
\idx & = & \encn(X^n)
\end{IEEEeqnarray} 
where 
\begin{equation}
  \label{eq:amos_def_enc}
 \encn \colon\mathcal{X}^n \to \idxset 
\end{equation}
is the encoding function and 
\begin{equation}
  \idxset \triangleq\{1,\ldots, \idxsetsize\}. 
\end{equation}
Based on the index $\idx$ and its side information~$Y^{n}$, the
decoder forms a reconstruction sequence
\begin{IEEEeqnarray}{rCl}\label{eq:decD}
\hat{X}_\d^n & = & \decn(\idx, Y^n)
\end{IEEEeqnarray}
where 
\begin{equation}
 \decn\colon \idxset \times \mathcal{Y}^n \to
\hat{\mathcal{X}}^n 
\end{equation}
is the decoder's reconstruction function.  The encoder's estimate of
the decoder's reconstruction sequence is
\begin{IEEEeqnarray}{rCl}\label{eq:decE}
\hat{X}_{e}^n & = & \ercn(X^n)
\end{IEEEeqnarray}
for some  
\begin{equation}
 \ercn\colon \mathcal{X}^n \to \hat{\mathcal{X}}^n. 
\end{equation}

The goal of the communication is that the decoder's reconstruction
$\hat{X}_{\d}^n$ matches the source sequence $X^n$ up to a distortion
no larger than $\Dd$ and the encoder's estimate $\hat{X}_{\e}^n$
matches the decoder's reconstruction $\hat{X}_{\d}^n$ up to a
distortion no larger than $\De$.  The distortions are measured by the
bounded, nonnegative, single-letter distortion functions $\dd(\cdot,\cdot)$ and $\de(\cdot,\cdot)$.

We say that a nonnegative triple $(R, \Dd, \De)$ is \emph{achievable} if for every
$\epsilon>0$ and sufficiently large $n$ there exists a message set of size 
\begin{equation}\label{eq:M}
|\set{M}| \leq 2^{n(R+\epsilon)}
\end{equation} and a triple of functions $(\encn, \decn, \ercn)$ as above  
such that  
the \emph{decoder-side reconstruction constraint}
\begin{equation}
\label{eq:distD}
\frac{1}{n} \sum_{i=1}^{n} \Exp \big[ \dd (X_i, \hat{X}_{\d,i}) \big]  \leq 
\Dd + \epsilon 
\end{equation}
and the \emph{encoder-side reconstruction constraint}
\begin{equation}
\label{eq:distE}  
\frac{1}{n} \sum_{i=1}^{n} \Exp \big[ \de (\hat{X}_{\d,i},\hat{X}_{\e,i}) \big]
\leq  \De+\epsilon
\end{equation}
are both met.

Our problem is not very interesting if the distortion constraints
cannot be met even when the source sequence is revealed losslessly to
the reconstructor. Consequently, we shall make the following
assumption throughout:
\begin{assumption}
  \label{assumption}
  The distortion functions $d_\d$ and $d_\e$ are such that for each
  $x\in\mathcal{X}$ there exist $\hat{x}_\d, \hat{x}_\e \in
  \hat{\set{X}}$ satisfying $d_\d(x, \hat{x}_\d)=0$
  and 
  $d_\e(\hat{x}_\d, \hat{x}_\e)=0$. 
  \end{assumption}
  As we shall see, this assumption ensures that the triple $(R, D_\d,
  D_\e)$ is achievable whenever $R\geq H(X|Y)$.

{
We are interested in finding the smallest rate $R$ such that a given distortion pair $D_\d, D_\e$ is achievable. For given $\Dd, \De \geq 0$, let  $\set{R}(D_\d, D_\e)$ denote the set of rates $R\geq 0$ such that the tuple $(R,D_\d, D_\e)$ is achievable:
\begin{equation}
\set{R}({D_\d, D_\e})\triangleq \{R \geq 0 \colon (R, D_\d, D_\e) \; \textnormal{is achievable}\}.
\end{equation}
Notice that by the assumption above, the set $\set{R}({D_\d, D_\e})$ contains all rates $R\geq H(X|Y)$ and is thus nonempty.
We can now define \emph{rate-distortions function} as
\begin{equation}\label{eq:rdfunct}
  \RopDD \triangleq \min_{R \in \set{R}({D_\d, D_\e})} R,
\end{equation}
where the minimum exists because the set $\set{R}{(D_\d, D_\e)}$ is nonempty, closed, and bounded from below by 0.}


\subsection{Related Setups}
\label{sec:related}


Wyner and Ziv's classic lossy source-coding problem with
side information~\cite{wynerziv76} is similar to our problem except
that Wyner and Ziv do not impose the encoder-side reconstruction
constraint~\eqref{eq:distE}. Informally, our problem thus reduces to
the Wyner-Ziv problem if we set $\De$ to infinity. Wyner and Ziv's
result can be summarized as follows:
 \begin{theorem}[Wyner and Ziv~\cite{wynerziv76}]\label{thm:wz}
   The rate-distortion function $R_{\textnormal{WZ}}(D_\d)$ in the
   Wyner-Ziv setup is given by
 \begin{equation}\label{eq:rdwz}
R_{\textnormal{WZ}}(D_\d)=\min_{ Z, \fdc} \bigl( I(X;Z)-I(Y;Z) \bigr)
\end{equation}
where $(X,Y) \sim P_{XY}$, and where the minimization is over all functions
$\fdc\colon \mathcal{Y} \times \mathcal{Z} \to \hat{\mathcal{X}}$ and 
discrete random variable $Z$  for which:
$Z$ takes values in an auxiliary alphabet
$\mathcal{Z}$ of size at most $|\mathcal{X}|+1$; 
\begin{equation}
Z\markov X\markov Y
\end{equation}
forms a Markov chain; and 
\begin{IEEEeqnarray}{rCl}
\Exp\big[d_\d\bigl(X, \fdc(Y,Z)\bigr)\big] & \leq & \Dd.
\end{IEEEeqnarray} 
\end{theorem}

Since imposing the encoder-side reconstruction
constraint~\eqref{eq:distE} cannot increase the set of achievable
rates,
\begin{equation}
  \label{eq:Us_geq_WZ}
  \RopDD \geq R_{\textnormal{WZ}}(D_\d).
\end{equation}
Equality holds whenever the encoder-side
reconstruction constraint~\eqref{eq:distE} does not pinch. 
For example, when $\hat{\mathcal{X}}=\mathcal{X}$;
$D_\d=D_\e$; and
 \begin{equation}
 d_\e(\hat{x},x) = d_\d(x,\hat{x}), \quad x,\hat{x} \in \mathcal{X}.
 \end{equation}
 Indeed, in this case the encoder can set $\hat{X}_{\e,i}$ to be
 $X_{i}$. This results in~\eqref{eq:distE} being identical to
 \eqref{eq:distD} and thus superfluous.

 Steinberg's setup in \cite{steinberg09} is obtained from ours by
 replacing the encoder-side distortion constraint~\eqref{eq:distE} by
 the more stringent perfect-reconstruction constraint
\begin{equation}\label{eq:steinberg}
\Prv{ \hat{X}_\e^n \neq \hat{X}_\d^n }\leq \epsilon.
\end{equation}
\begin{theorem}[Steinberg~\cite{steinberg09}]\label{thm:steinberg}
The rate-distortion function $R_{\textnormal{cr}}(D_\d)$ in
Steinberg's setup is given by 
 \begin{equation}
   \label{eq:amos_ayef200}
  R_\textnormal{cr}(D_{\mathrm{d}}) \triangleq
  \min_{\hat{X}}
  \big( I(X;\hat{X}) - I(Y;\hat{X}) \big),
\end{equation}
where the minimization is over all $\hat{X}$ taking value in $\hat{\mathcal{X}}$ and satisfying
\begin{equation}
\hat{X} \markov X \markov Y
\end{equation}
 and
 \begin{equation}
 \E{d_{\mathrm{d}}(X,\hat{X})} \leq D_{\mathrm{d}}.
 \end{equation}
 \end{theorem}
 \begin{remark}\label{rem:steinberg}
 Constraint~\eqref{eq:steinberg} is equivalent to the block-distortion constraint 
 \begin{equation}\label{eq:steinberg2}
 \E{ \textnormal{I}\{ \hat{X}_\e^n \neq \hat{X}_{\d}^n\}} \leq\epsilon.
 \end{equation}
 Thus, when in our setup $d_\e(\cdot,\cdot)$ is the Hamming distortion
 and $D_\e=0$, then Steinberg's setup differs from ours only in that
 \eqref{eq:steinberg} is a block-distortion constraint whereas
 \eqref{eq:distE} is an average-per-symbol distortion constraint.
\end{remark}



\subsection{Results}
To describe the rate-distortions function for the setup of
Section~\ref{subsec:setup}, we introduce the function $\RitDD$.  The
expression for $\RitDD$ in 
is similar to the expression for $R_{\textnormal{WZ}}(D_\d)$
in~\eqref{eq:rdwz} except that in the expression for $\RitDD$ we have
the additional constraint; see~\eqref{eq:distEr0} ahead.

Given the joint law $P_{XY}$ of the source and side information, and
given the distortion functions $d_{\d}, d_{\e}$, this function is
defined as
\begin{equation}
\label{eq:rdf}
\RitDD=\min_{ Z, \fdc, \fec} \bigl( I(X;Z)-I(Y;Z) \bigr)
\end{equation}
where the minimization is over all discrete random variables~$Z$
taking value in some finite auxiliary alphabet $\mathcal{Z}$
and forming the Markov chain
\begin{equation}
\label{eq:amos400}
Z \markov X \markov Y
\end{equation} 
and over the functions $\fdc\colon \mathcal{Y} \times \mathcal{Z} \to
\hat{\mathcal{X}}$ and $\fec\colon \mathcal{X} \times \mathcal{Z} \to
\hat{\mathcal{X}}$ satisfying
\begin{IEEEeqnarray}{rCl}
\Exp\big[d_\d\bigl(X, \fdc( Y,Z)\bigr)\big] & \leq & \Dd \label{eq:distDr0}\\ 
\Exp\big[d_\e\bigl(\fdc(Y,Z),\fec(X,Z)\bigr)\big] & \leq & \De.
\label{eq:distEr0}
\end{IEEEeqnarray}

Note that, thanks to Assumption~\ref{assumption}, the feasible set
in~\eqref{eq:rdf} is not empty: we can choose $Z$ as $X$ and $\fdc$,
$\fec$ as the functions whose existence is guaranteed by the
assumption. This choice demonstrates that
\begin{equation}
  \label{eq:amos_bounded}
  \RitDD \leq H(X|Y).
\end{equation}

Using the convex cover method \cite{ElGamalKim2011} it
can be shown that:
\begin{remark}
  Allowing for sets $\mathcal{Z}$ of cardinality greater than
  $|\mathcal{X}| + 3$ does not decrease the value of the optimization
  problem.
\end{remark}

A consequence of this remark is that the minimum in \eqref{eq:rdf}
is achieved: indeed, we may choose $\set{Z}$ as the set $\{1, \ldots,
|\mathcal{X}| + 3\}$ with result that there are only a finite number
of functions $\fdc$, $\fec$, and the problem is reduced to minimizing a
continuous function over a compact set.

The key properties of $\RitDD$ are summarized in the following
proposition:
\begin{proposition}[Key Properties of the Function $\RitDD$]
  \label{prop:key_properties}
  The function $\RitDD\colon \RealsP^{2} \to \RealsP$ is bounded from
  above by $H(X|Y)$ and is nondecreasing in the distortions
\begin{multline*}
\Bigl( \Dd' \geq \Dd \; \text{and} \; \De' \geq \De \Bigr) 
\Rightarrow
 \Bigl( \Rit{\Dd'}{\De'} \leq \RitDD \Bigr).
\end{multline*}
Moreover, it is convex and continuous.
\end{proposition}
\begin{proof} 
See Appendix~\ref{app:key_properties}.
\end{proof}
Our main result can be now stated as:
\begin{theorem}
\label{thm:finite}
The rate-distortions function for the setup in
Section~\ref{subsec:setup} is equal to $\RitDD$
\begin{equation}
  \label{eq:amos_main100}
  \RopDD = \RitDD.
\end{equation}
 \end{theorem}

\begin{proof}[Proof of Theorem~\ref{thm:finite}]
  The coding scheme that establishes achievability is a variation on
  the coding scheme of Wyner and Ziv~\cite{wynerziv76} and is thus
  only sketched. Its analysis is omitted.

  Fix $Z, \phi, \psi$ satisfying~\eqref{eq:amos400} and
  \eqref{eq:distEr0}, and fix also a blocklength~$n$ and some (small)
  $\epsilon>0$. Let $\set{C}$ be a random blocklength-$n$ codebook
  with $\lfloor 2^{ n (I(X;Z)- I(Y;Z) +2\epsilon)}\rfloor$ bins, each
  containing approximately $2^{ n ( I(Y;Z) -\epsilon)}$ codewords with
  the total number of codewords thus being $\lfloor 2^{n
    (I(X;Z)+\epsilon)}\rfloor$.
%
  Generate the codewords independently with the components of each
  codeword being drawn IID $P_Z$. Number the bins $1$ through $\lfloor
  2^{ n (I(X;Z)- I(Y;Z) +2\epsilon)}\rfloor$.

  Upon observing the source sequence $X^n$, the encoder seeks a
  codeword $Z^{*n}$ in $\mathcal{C}$ that is jointly typical with
  $X^n$. If successful, it sends the number of the bin containing
  $Z^{*n}$ as the message~$M$. It also 
  produces the reconstruction sequence~$\hat{X}_\e^n$ by applying the
  function $\psi$ componentwise to $Z^{*n}$ and $X^n$. The decoder
  seeks a codeword $\hat{Z}^n$ in Bin~$M$ that is jointly typical
  with its side-information $Y^n$ and applies the reconstruction
  function $\phi$ componentwise to $\hat{Z}^n$ and $Y^n$ to produce
  $\hat{X}_\d^n$.

  The converse is proved in Subsection~\ref{sec:th1}.  
\end{proof}

Though not identical, Steinberg's setup is very similar to our setup
when $\d_\e(\cdot,\cdot)$ is the Hamming distortion and
$D_{\e}$ is zero (Remark~\ref{rem:steinberg}). It is therefore
  not surprising that, as the following corollary shows, the two
  setups lead to identical rates:
\begin{corollary}\label{cor:steinberg}
  Let $d_\d(\cdot,\cdot)$ be arbitrary, and let $d_\e(\cdot,\cdot)$ be
  the Hamming distortion function
\begin{equation}
  d_{\mathrm{e}}(\hat{x}_{\mathrm{d}},\hat{x}_{\mathrm{e\vphantom{d}}})
  =  \; \mathrm{I}\{\hat{x}_{\mathrm{d}}
  \neq \hat{x}_{\mathrm{e\vphantom{d}}}\}, \quad \hat{x}_{\mathrm{d}}, \hat{x}_{\mathrm{e\vphantom{d}}} \in \mathcal{\hat{X}}.
\end{equation}
Then
 \begin{equation}
R(D_{\d},D_{\e})\Big|_{D_\e = 0} = R_{\textnormal{cr}}(D_\d).
\end{equation}
\end{corollary}
\begin{proof} See Appendix~\ref{app:RDcr}.
\end{proof}

\begin{remark}Our results can be extended to a scenario where the encoder observes
not only the source sequence $\{X_{i}\}$ but also some sequence
$\{W_{i}\}$ which is correlated with the decoder's side-information
sequence $\{Y_{i}\}$. This additional sequence $\{W_{i}\}$ makes it
easier for the encoder to estimate the decoder's reconstruction
sequence and thus allows the decoder to rely more heavily on its side
information $\{Y_{i}\}$. To see how this seemingly more general
scenario reduces to our scenario assume that $\{(X_i, W_i,
Y_i)\}_{i=1}^{n}$ are IID random triples of law $P_{XWY}$ and that
$W_{i}$ takes value in the finite set $\mathcal{W}$. Consider now a
new IID source $\{\tilde{X}_{i}\}$ taking value in the set
$\tilde{\mathcal{X}} = \mathcal{X} \times \mathcal{W}$ according to
the law $P_{XW}$ with $\tilde{X}_{i} = (X_{i}, W_{i})$. The encoder
now observes the source sequence $\{\tilde{X}_{i}\}$ only and no
additional sequences. The decoder side information is still
$\{Y_{i}\}$, and the joint law of $\tilde{X}_{i}, Y_{i}$ is
$P_{XWY}$. Finally define the new decoder distortion function
$\tilde{d}_{\text{d}} \colon \tilde{\mathcal{X}} \times
\hat{\mathcal{X}} \to \mathbb{R}^{+}$ as
  \begin{equation*}
  \tilde{d}_{\text{d}}\bigl( (X_{i}, W_{i}), \hat{X}_{i} \bigr) =
    \dd(X_{i}, \hat{X}_{i}),
\end{equation*}
i.e., the distortion function $\tilde{d}_{\d}$ does not depend on the $W_i$-component.
Solving the original scenario for this new source and new decoder
distortion function is equivalent to solving the seemingly more general
problem we described.
\end{remark}

\subsection{Proof of the Converse to Theorem~\ref{thm:finite}}\label{sec:th1}

To establish the converse, we show that if a triple $(R,
\Dd, \De)$ is achievable, then for every $\epsilon > 0$
\begin{equation}\label{eq:cv-to-show}
R+\epsilon \geq \Rit{\Dd+\epsilon}{\De+\epsilon}.
\end{equation}
Since $\RitDD$ is continuous (Proposition~\ref{prop:key_properties}),
and since $\epsilon$ can be arbitrarily small, this implies that $R
\geq \RitDD$ whenever $(R, \Dd, \De)$ is achievable, and consequently
that $\RopDD \geq \RitDD$.

The first part of our proof identifying the auxiliary random variable
$Z_{i}$ \eqref{eq:Zi} and the function $\deci$ \eqref{eq:def1} is
similar to the proof of the Wyner-Ziv result
\cite{ElGamalKim2011}. For a given blocklength-$n$ code $\encn$,
$\decn$, $\ercn$ satisfying
\eqref{eq:M}--\eqref{eq:distE}, 
we have \setcounter{myieq}{1}
\begin{IEEEeqnarray}{rCl}
  \IEEEeqnarraymulticol{3}{l}{
  n(R+\epsilon)} \nonumber \\
  \; & \myieq{\geq} & H(\idx) \label{eq:25}\\
     & \myieq{\geq} & I(X^n;\idx|Y^n) \\
     & \myieq{=}    & \sum_{i = 1}^n I(X_i;\idx|Y^n,X^{i-1}) \\
     &{=}    & \sum_{i = 1}^n H(X_i|Y^n,X^{i-1}) -
                                       H(X_i|\idx,Y^n,X^{i-1}) \\
     & \myieq{=}    & \sum_{i = 1}^n H(X_i|Y_i) -
                                       H(X_i|\idx,Y^n,X^{i-1}) \\
     & \myieq{\geq} & \sum_{i = 1}^n H(X_i|Y_i) -
                                       H(X_i|\idx,Y^n) \\
     & \myieq{=}    & \sum_{i = 1}^n H(X_i|Y_i) -
                                       H(X_i|Z_i,Y_i) \\
     & =& \sum_{i = 1}^n I(X_i;Z_i|Y_i) \\
     & \myieq{=}    & \sum_{i = 1}^n H(Z_i|Y_i) -
                                       H(Z_i|X_i) \\
     &{=}    & \sum_{i = 1}^n I(X_i;Z_i) - I(Y_i;Z_i),
                      \label{equ:cv-part1-last}
\end{IEEEeqnarray}
where \setcounter{myieq}{1}
\myieqexp follows by \eqref{eq:M};
\myieqexp follows because  conditioning cannot increase entropy and because $H(M|Y^n, X^n)\geq 0$;
\myieqexp follows from the chain rule for mutual information;
\myieqexp follows because the pair $X_i, Y_i$ is independent of the tuple $(X_1^{i-1}, Y_1^{i-1}, Y_{i+1}^n)$; 
\myieqexp follows from the fact that conditioning cannot increase entropy;
\myieqexp follows by defining
\begin{equation}\label{eq:Zi}
  Z_i \triangleq (\idx,Y^{i-1},Y_{i+1}^n);
\end{equation}
and \myieqexp follows  because with the definition above
\begin{equation}
\label{eq:amos200}
Z_i \markov X_i \markov Y_i.
\end{equation}

Denote by $\phi_i^{(n)}$ the function that maps $(M,Y^n)$ to the
$i$-th component of the $n$-tuple $\phi^{(n)}(M,Y^n)$, and denote
by~$\psi_i^{(n)}$ the function that maps $X^n$ to the $i$-th component
of the $n$-tuple $\psi^{(n)}(X^n)$. Since there is a one-to-one
correspondence between the pairs $(Y_{i},Z_{i})$ and $(M,Y^{n})$, we
can define a function $\phi_{i}$ that maps $(Y_{i},Z_{i})$ to
$\phi_{i}^{(n)}(M,Y^{n})$
\begin{equation}\label{eq:def1}
  \deci(Y_i,Z_i) \triangleq \decni(\idx,Y^n).
\end{equation}
We now define
\begin{equation}\label{eq:def2}
  \Ddi \triangleq
  \E{\dd\bigl(X_i,\decni(\idx,Y^n)\bigr)},
\end{equation}
where $\Exp[\cdot]$ is with respect to
$P_{X_{\vphantom{I}}^nY_{\vphantom{I}}^n}$.
By definitions \eqref{eq:def1} and \eqref{eq:def2}, 
\begin{equation}\label{equ:cv-exp-dd-deci-ddi}
  \BigE{\dd\bigl(X_i,\deci(Y_i,Z_i)\bigr)} = \Ddi,
\end{equation}
where $\Exp[\cdot]$ is with respect to
$P_{X_iY_i\vphantom{|}}P_{Z_i|X_i}$.

We next turn to the encoder-side distortion. We will show that there
exists a deterministic function $\erci\colon \mathcal{X} \times
\mathcal{Z} \to \mathcal{\hat{X}}$ that achieves a distortion no
larger than $\Dei$, where $\Dei$ is the distortion achieved by
$\ercni(X^n)$, namely,
\begin{equation}
  \Dei \triangleq
  \E{\de\bigl(\decni(\idx,Y^n),\ercni(X^n)\bigr)}.
\end{equation}

To this end, we express $\Dei$ as
\newcommand{\Xnminusi}{{X}_{\! \backslash\!i}}
\newcommand{\xnminusi}{{x}_{\!\backslash \!i}}
\newcommand{\Xnminusistar}{{X}_{\!\backslash \!i}^*}
\newcommand{\xnminusistar}{{x}_{\!\backslash \!i}^*}
\begin{IEEEeqnarray}{rCl}
  \IEEEeqnarraymulticol{3}{l}{
    \Dei }\nonumber\\ \quad
  & = & \BigE[X^n,Y_i,Z_i]{\de\bigl(\deci(Y_i,Z_i),\ercni(X^n)\bigr)} \\
  & = & \Exp_{X^n,Z_i} \BigE[Y_i|X^n,Z_i]{\de\bigl(\deci(Y_i,Z_i),\ercni(X^n)\bigr)} \\
  & = & \Exp_{X^n,Z_i}
   \BigE[Y_i|X_i,\Xnminusi,Z_i]{\de\bigl(\deci(Y_i,Z_i),\ercni(X_i,\Xnminusi)\bigr)},
       \IEEEeqnarraynumspace \label{equ:cv-exp-de-xnminusi}
\end{IEEEeqnarray}
where $\Xnminusi \triangleq (X^{i-1},X_{i+1}^n)$.
For every $(x_i,z_i) \in \alphSRC \times \alphAUX$, we define $\xnminusistar(x_i,z_i)$ (or for short $\xnminusistar$) as:\footnote{If $\operatorname*{arg\,min}$ is not unique, $\xnminusi(x_i, z_i)$ is defined as the
first in lexicographical order.} 
\begin{multline} \label{eq:xstar1}
  \xnminusistar(x_i,z_i) \triangleq
  \operatorname*{arg\,min}_{\xnminusi \in \alphSRC^{n-1}} \\
  \BigE[Y_i|X_i=x_i,\Xnminusi=\xnminusi,Z_i=z_i]{\de \bigl(
    \deci(Y_i,z_i),\ercni(x_i,\xnminusi) \bigr) }
\end{multline}

or in any other way that guarantees
\begin{multline}
\Exp_{\Xnminusi|X_i=x_i,Z_i=z_i} \\ 
\BigE[Y_i|X_i=x_i,\Xnminusi,Z_i=z_i]
{\de\bigl(\deci(Y_i,z_i),\ercni(x_i,\Xnminusi)\bigr)}
\geq \\
\BigE[Y_i|X_i=x_i,\Xnminusi=\xnminusistar,Z_i=z_i]{\de\bigl(\deci(Y_i,z_i),\ercni(x_i,\xnminusistar)\bigr)}.
\end{multline}
We can now define the function $\erci$ as
\begin{subequations}
\label{block:amos_gym}
\begin{IEEEeqnarray}{rCl}
  \erci \colon \alphSRC \times \alphAUX & \to & \alphREC \\
                              (x_i,z_i) & \mapsto & \ercni\bigl(x_i,\xnminusistar(x_i,z_i)\bigr).
\end{IEEEeqnarray}
\end{subequations}
For every $(x_i,\xnminusi,z_i) \in \alphSRC^n \times \alphAUX$, we have
\setcounter{myieq}{1}
\begin{IEEEeqnarray}{rCl}
  \IEEEeqnarraymulticol{3}{l}{
  \BigE[Y_i|X_i=x_i,\Xnminusi=\xnminusi,Z_i=z_i]{\de
    \bigl(\deci(Y_i,z_i),\ercni(x_i,\xnminusi) \bigr)}
  } \nonumber \\
  \;\, & \myieq{\geq} &
  \BigE[Y_i|X_i=x_i,\Xnminusi=\xnminusistar,Z_i=z_i]{\de
    \bigl(\deci(Y_i,z_i),\ercni(x_i,\xnminusistar) \bigr)}
  \IEEEeqnarraynumspace\\
  & \myieq{=} & 
  \BigE[Y_i|X_i=x_i,Z_i=z_i]{\de
    \bigl(\deci(Y_i,z_i),\ercni(x_i,\xnminusistar) \bigr)}\\
  & \myieq{=} & 
  \BigE[Y_i|X_i=x_i,Z_i=z_i]{\de \bigl(\deci(Y_i,z_i),\erci(x_i,z_i)\bigr)},
  \label{equ:cv-cond-exp-de-erci}
\end{IEEEeqnarray}
where \setcounter{myieq}{1}
\myieqexp follows from the definition of $\xnminusistar$;
\myieqexp follows because
\begin{equation}
  \Xnminusi \markov (X_i,Z_i) \markov Y_i;
\end{equation}
and \myieqexp follows from the definition of~$\erci$ \eqref{block:amos_gym}.

It now follows from~\eqref{equ:cv-exp-de-xnminusi}
and~\eqref{equ:cv-cond-exp-de-erci} that
%
\begin{equation}\label{equ:cv-exp-de-erci-dei}
  \BigE[X_i,Y_i,Z_i]{\de\bigl(\deci(Y_i,Z_i),\erci(X_i,Z_i)\bigr)} \leq \Dei.
\end{equation}

Continuing from~(\ref{equ:cv-part1-last}) we thus obtain
\setcounter{myieq}{1}
\begin{IEEEeqnarray}{rCl}
  n(R+\epsilon) & \geq &
  \sum_{i = 1}^n I(X_i;Z_i) - I(Y_i;Z_i)
 \\
& \myieq{\geq} &
    \sum_{i = 1}^n \Ritsym(\Ddi,D_{\mathrm{e\vphantom{d}},i}) 
    \label{equ:cv-part2-first}\\
  & \myieq{=} & n \frac{1}{n}
    \sum_{i = 1}^n \Ritsym(\Ddi,D_{\mathrm{e\vphantom{d}},i}) \\
  & \myieq{\geq} & 
    n \Ritsym\bigg(\frac{1}{n} \sum_{i = 1}^n \Ddi\,,
        \frac{1}{n} \sum_{i = 1}^n D_{\mathrm{e\vphantom{d}},i}\bigg) \\
  & \myieq{\geq} & n \Rit{\Dd+\epsilon}{\De+\epsilon}
  \label{eq:amos300}
\end{IEEEeqnarray}
where \setcounter{myieq}{1} \myieqexp follows from the definition of
$\RitDD$ and from \eqref{eq:amos200}, \eqref{equ:cv-exp-dd-deci-ddi}, and
\eqref{equ:cv-exp-de-erci-dei};
\myieqexp follows by multiplying by $1$; \myieqexp follows from
the convexity of $\RitDD$ (Proposition~\ref{prop:key_properties});
and \myieqexp follows from the monotonicity of $\RitDD$
(Proposition~\ref{prop:key_properties}) and the fact that 
$\frac{1}{n} \sum_{i = 1}^n \Ddi\leq \Dd + \epsilon$
and
$\frac{1}{n} \sum_{i = 1}^n \Dei\leq \De + \epsilon$.
This establishes~\eqref{eq:cv-to-show} and thus concludes the proof of the converse.

%
%
%

\section{Gaussian Source and Quadratic Distortions}\label{sec:Gauss}
\subsection{Setup}\label{sec:Gausssetup}
We next consider the case where the source, side information, and
reconstruction alphabets $\mathcal{X}, \mathcal{Y}, \hat{\mathcal{X}}$
are the reals $\mathbb{R}$; the distortion functions $d_\d$ and $d_\e$ are quadratic
\begin{IEEEeqnarray}{rCl}
 d_{\d}(x,
  \hat{x}_\d)&=&(x-\hat{x}_\d)^2,\label{eq:dd}\\
  d_\e(\hat{x}_\d, \hat{x}_\e) &=&
  (\hat{x}_\d-\hat{x}_\e)^2;\label{eq:de}
  \end{IEEEeqnarray}
  and the source and side-information pair $(X,Y)$ is a centered
  bivariate Gaussian, where $X$ is of variance $\sigma_X^2$
  \begin{equation}
    \sigma_{X} > 0
  \end{equation}
  and $Y=\xi X+U$ for some centered Gaussian $U$ that is independent
  of~$X$ and that is of variance $\sigma_U^2$ and where $\xi$ is a
  nonzero constant.\footnote{The problem is not interesting when $\xi$
    is zero, because in this case the side information is independent
    of the source and is thus irrelevant.}  The rate-distortions
  function depends on $\xi$ only through the ratio $\sigma_U^2/\xi^2$,
  because the receiver can premultiply its side information by
  $\xi^{-1}$ without affecting the rate-distortions function. In the
  following we thus assume that $\xi=1$, i.e.,
\begin{equation}\label{eq:Y}
  Y=X+U.
\end{equation}
We denote the rate-distortions function for this setup by
$R^{\textnormal{G}}(D_\d, D_\e)$.

When $\sigma_{U}$ is zero the problem is not interesting, because in
this case the source sequence is determined by the side information,
and $R^{\textnormal{G}}(D_\d, D_\e)$ is thus zero for all nonnegative
values of $D_\d$ and $D_\e$. We shall henceforth thus assume
\begin{equation}
  \sigma_{U} > 0.
\end{equation}
In this case, no finite rate can allow $D_\d$ to be zero (even if we ignore
the encoder-side reconstruction constraint). Thus, we shall also
assume
\begin{equation}
D_\d > 0.  
\end{equation}

%

\subsection{Related Work}
As we have seen in Section~\ref{sec:related}, the Wyner-Ziv setup is
obtained from ours if the encoder-side reconstruction
constraint~\eqref{eq:distE} is omitted, and Steinberg's common
reconstruction setup is obtained if~\eqref{eq:distE} is replaced
by~\eqref{eq:steinberg}.

For a Gaussian source and quadratic distortion measures, Steinberg's
common reconstruction rate-distortion function is \cite{steinberg09}
 \begin{equation}\label{eq:RcrG}
   R_{\textnormal{cr}}^{\textnormal{G}}(D_\d) =  
   \frac{1}{2} \logplus 
   \frac{\sigma_X^2(\sigma_U^2 + D_{\mathrm{d}})}
   {(\sigma_X^2 + \sigma_U^2)D_{\mathrm{d}} }.
\end{equation}

The Wyner-Ziv rate-distortion function is \cite{wynerziv76}
\begin{equation}\label{eq:RWZG}
R_{\textnormal{WZ}}^{\textnormal{G}}(D_\d) = 
  \frac{1}{2} \logplus \frac{\sigma_X^2 \sigma_U^2}{(\sigma_X^2 +
  \sigma_U^2)D_{\mathrm{d}}}.
\end{equation}
This is the rate-distortion function even if the side information is
revealed not only to the decoder but also to the encoder.  

\subsection{Result}\label{sec:resGaussian}

\begin{theorem}\label{thm:gaussian}
  For a Gaussian source and quadratic distortion measures, the
  rate-distortions function $R^{\textnormal{G}}(D_\d, D_\e)$ can be
  expressed as follows:

  If $\sqrt{\De \sigma_U^2} \geq \min\left\{ \Dd, \frac{\sigma_X^2
      \sigma_U^2}{\sigma_X^2 +\sigma_U^2} \right\}$,
  then\vspace{-0.12cm}
  \begin{equation*}
    R^{\textnormal{G}}(D_\d, D_\e) = 
    \frac{1}{2} \logplus \frac{\sigma_X^2 \sigma_U^2}{(\sigma_X^2 +
        \sigma_U^2)D_{\mathrm{d}}}.
  \end{equation*}

  If $\sqrt{\De \sigma_U^2} < \min\left\{ \Dd, \frac{\sigma_X^2
      \sigma_U^2}{\sigma_X^2 +\sigma_U^2}\right\}$,
  then\vspace{-0.11cm}
\begin{IEEEeqnarray*}{lCl}
  \label{eq:RopGa2}
 R^{\textnormal{G}}(D_\d, D_\e) &   = & 
  \frac{1}{2} \logplus \! \left(\frac{\sigma_X^2}{\sigma_X^2 + \sigma_U^2}
  \frac{\sigma_U^2 + D_{\mathrm{d}} - 2 \sqrt{\sigma_U^2
  D_{\mathrm{e\vphantom{d}}}}}{D_{\mathrm{d}} -
  D_{\mathrm{e\vphantom{d}}}}\right).
  \end{IEEEeqnarray*}
  \end{theorem}
  \begin{proof}
The direct part is proved in Section~\ref{sec:achGauss} and the converse in Section~\ref{sec:convGauss}.
  \end{proof}
\begin{remark}
  If $D_\e=0$, then our rate-distortions function
  $R^{\textnormal{G}}(D_\d,0)$ coincides with Steinberg's
  common-reconstruction rate-distortion function
  $R_{\textnormal{cr}}^{\textnormal{G}}(D_\d)$ of~\eqref{eq:RcrG}:
\begin{equation}
  R^{\textnormal{G}}(D_{\d},D_{\e})\Big|_{D_\e = 0} = 
  R^{\textnormal{G}}_{\textnormal{cr}}(D_\d).
\end{equation}
\end{remark}
\begin{remark}
  \label{rem:wzG}
  If $D_{\d}$ and $D_{\e}$ are such that
  \begin{equation}\label{eq:si1}
    \sqrt{\De \sigma_U^2} \geq \min\left\{\Dd, \frac{\sigma_X^2
        \sigma_U^2}{\sigma_X^2+\sigma_U^2}\right\}
  \end{equation}
  or 
  \begin{equation}\label{eq:si2}
    \left(1- \sqrt{\frac{D_\e}{\sigma_U^2}}\right)^2 \sigma_X^2 \leq D_\d-D_\e
  \end{equation} 
  then $R^{\textnormal{G}}(D_\d,D_\e)$ coincides with Wyner and Ziv's
  rate-distortion function
  $R_{\textnormal{WZ}}^{\textnormal{G}}(D_\d)$
  in~\eqref{eq:RWZG}. Thus, if \eqref{eq:si1} or \eqref{eq:si2} holds,
  then relaxing Constraint~\eqref{eq:distE} and/or revealing the
  side information also to the encoder does not decrease the
  rate-distortions function.
\end{remark}

\subsection{The Direct Part of Theorem~\ref{thm:gaussian}} \label{sec:achGauss}

In the two cases that we shall describe in~\eqref{eq:d3}
and~\eqref{eq:d4} ahead, no encoding is necessary because the encoder
and the decoder can produce sufficiently good reconstructions
$\hat{X}_\e^n$ and $\hat{X}_\d^n$ based solely on their observed
sequences $X^n$ and $Y^n$. In these cases
$R^{\textnormal{G}}(D_\d,D_\e)$ is thus zero.
\begin{enumerate}
  \item[1)] If 
  \begin{subequations}\label{eq:d3}
  \begin{equation}\sqrt{D_{\mathrm{e\vphantom{d}}} \sigma_U^2} \geq
           \min\left\{D_{\mathrm{d}},
             \frac{\sigma_X^2 \sigma_U^2}
                   {\sigma_X^2 + \sigma_U^2}
           \right\}
           \end{equation}
           and 
           \begin{equation}
           D_{\mathrm{d}} \geq
               \frac{\sigma_X^2 \sigma_U^2}{\sigma_X^2 + \sigma_U^2},
               \end{equation}
               \end{subequations}
             then the encoder and  decoder can produce the sequences
             \begin{IEEEeqnarray}{rCl}
             \hat{X}_\e^n &= &\frac{\sigma_X^2}{\sigma_X^2 + \sigma_U^2} X^n\\
               \hat{X}_\d^n &= &\frac{\sigma_X^2}{\sigma_X^2 + \sigma_U^2} Y^n
             \end{IEEEeqnarray}
which satisfy the distortion constraints.

  \item[2)]  If  
  \begin{subequations}\label{eq:d4}
  \begin{equation}
  \sqrt{D_{\mathrm{e\vphantom{d}}} \sigma_U^2} <
           \min\left\{D_{\mathrm{d}},
             \frac{\sigma_X^2 \sigma_U^2}
                   {\sigma_X^2 + \sigma_U^2}
           \right\}
           \end{equation}
           and 
           \begin{equation}
           D_{\mathrm{d}} \geq 
               \sigma_X^2 \left(1 -
               \sqrt{\frac{D_{\mathrm{e}}}{\sigma_U^2}}\right)^2 + D_{\mathrm{e\vphantom{d}}},
               \end{equation}
               \end{subequations}
                    then the encoder and  decoder can produce the sequences
             \begin{IEEEeqnarray}{rCl}\label{eq:Xe1}
             \hat{X}_\e^n &= & \sqrt{\frac{D_{\mathrm{e\vphantom{d}}}}{\sigma_U^2}} X^n\\
               \hat{X}_\d^n &= & \sqrt{\frac{D_{\mathrm{e\vphantom{d}}}}{\sigma_U^2}} Y^n\label{eq:Xd1}
             \end{IEEEeqnarray}
which satisfy the distortion constraints.
 \end{enumerate}

 The achievability of Theorem~\ref{thm:gaussian} in the remaining
 cases will be established using the following proposition with a
 judicious choice of the parameters.
\begin{proposition}\label{prop:Gach} 
  For the setup in Section~\ref{sec:Gausssetup} of a Gaussian source
  and quadratic distortion measures, the tuple $(R, D_\d, D_\e)$ is
  achievable whenever
\begin{IEEEeqnarray}{rCl}
R \geq   \frac{1}{2}
    \log
    \frac{\sigma_X^2 \sigma_U^2 + \sigma_X^2 \sigma_W^2 + \sigma_U^2 \sigma_W^2}
         {(\sigma_X^2 + \sigma_U^2)\sigma_W^2}       
        \label{equ:gaussian-ac-diff-i-log}
  \end{IEEEeqnarray}
 for some parameters $\sigma_W^2,a>0$ and $b \geq 0$ satisfying
 \begin{subequations}
   \label{block:amos_lunch}
  \begin{equation}\label{equ:gaussian-ac-scheme-condition-dec}
  (1-a-b)^2 \sigma_X^2
              + a^2 \sigma_W^2
              + b^2 \sigma_U^2 \leq D_{\mathrm{d}}
  \end{equation}
  and
  \begin{equation}\label{equ:gaussian-ac-scheme-condition-enc}
    b^2 \sigma_U^2 \leq D_{\mathrm{e\vphantom{d}}}.
  \end{equation}
  \end{subequations}
  Thus, 
  \begin{equation}\label{eq:Rgup}
  R^{\textnormal{G}}(D_\d, D_\e) \leq \min_{a,\, b, \, \sigma_W^2\;}   \frac{1}{2}
    \log 
    \frac{\sigma_X^2 \sigma_U^2 + \sigma_X^2 \sigma_W^2 + \sigma_U^2 \sigma_W^2}
         {(\sigma_X^2 + \sigma_U^2)\sigma_W^2},
  \end{equation}
  where the minimization is over all $\sigma_W^2,a>0$ and $b \geq 0$
  satisfying~\eqref{block:amos_lunch}.
\end{proposition}
\begin{proof} 
See Appendix~\ref{sec:Gscheme}.
\end{proof}

We can now prove the achievability part of Theorem~\ref{thm:gaussian}
for the remaining cases.
\begin{enumerate}
\item[3)] If
\begin{subequations}\label{eq:d1} 
\begin{equation}\sqrt{D_{\mathrm{e\vphantom{d}}} \sigma_U^2} \geq
           \min\left\{D_{\mathrm{d}},
             \frac{\sigma_X^2 \sigma_U^2}
                   {\sigma_X^2 + \sigma_U^2}
           \right\} 
          \end{equation}
           and
           \begin{equation} 
             \label{eq:amos_aldi10}
           D_{\mathrm{d}} <
               \frac{\sigma_X^2 \sigma_U^2}{\sigma_X^2 + \sigma_U^2},
               \end{equation} 
               \end{subequations}
                then the choice
                \begin{subequations}\label{eq:c1}
                  \begin{equation}
                    \sigma_W^2   =  \frac{D_{\mathrm{d}}}{1 - \frac{\sigma_X^2 +
                        \sigma_U^2}{\sigma_X^2 \sigma_U^2} D_{\mathrm{d}}} 
                  \end{equation}
                  (which is positive by~\eqref{eq:amos_aldi10}) and
                \begin{eqnarray}
  a     & = & \frac{D_{\mathrm{d}}}{\sigma_W^2} \nonumber \\ 
  &=& 1-\frac{\sigma_X^2 +\sigma_U^2}{\sigma_X^2\sigma_U^2} D_\d, \\
  b    & = & \frac{\sigma_X^2}{\sigma_X^2 + \sigma_U^2} (1 - a) \nonumber \\ 
  &=& \frac{D_{\mathrm{d}}}{\sigma_U^2}. 
\end{eqnarray}
\end{subequations}
satisfies \eqref{block:amos_lunch} because
\begin{align}
  & (1 - a- b)^2 \sigma_X^2 + a^2\sigma_W^2 + b^2 \sigma_U^2 \nonumber \\ & \qquad = 
  \left(\frac{\sigma_X^2 + \sigma_U^2}{\sigma_X^2} b - b\right)^2 \sigma_X^2 + \frac{D_{\mathrm{d}}^2}{\sigma_W^2} +    \frac{D_{\mathrm{d}}^2}{\sigma_U^2}\\
    & \qquad = \frac{D_{\mathrm{d}}^2}{\sigma_X^2} + D_{\mathrm{d}} \left(1 - \frac{\sigma_X^2 +
  \sigma_U^2}{\sigma_X^2 \sigma_U^2} D_{\mathrm{d}}\right) +
    \frac{D_{\mathrm{d}}^2}{\sigma_U^2}\\
    & \qquad = D_{\mathrm{d}}\label{eq:Dd2}
\end{align}
and
\begin{equation}
    b^2  \sigma_U^2
     =  \frac{D_{\mathrm{d}}^2}{\sigma_U^2}
     \leq D_{\mathrm{e\vphantom{d}}}. 
\end{equation}
Moreover, for this choice,
\begin{IEEEeqnarray}{rCl}\label{eq:rates2}
\lefteqn{\frac{1}{2}
    \log
    \frac{\sigma_X^2 \sigma_U^2 + \sigma_X^2 \sigma_W^2 + \sigma_U^2 \sigma_W^2}
         {(\sigma_X^2 + \sigma_U^2)\sigma_W^2}       
    } \qquad \nonumber \\ &    =&  \frac{1}{2} \log \frac{\sigma_X^2
                \sigma_U^2}{(\sigma_X^2 + \sigma_U^2)D_{\mathrm{d}}}.
\end{IEEEeqnarray}
Thus, by~\eqref{eq:Dd2}--\eqref{eq:rates2} and by Proposition~\ref{prop:Gach}, we conclude that when $D_\d$ and $D_\e$ satisfy~\eqref{eq:d1}, 
  \begin{equation}\label{equ:gaussian-ac-rate-wz}
             R^\textnormal{G}(D_\d, D_\e) \leq   \frac{1}{2} \log \frac{\sigma_X^2
                \sigma_U^2}{(\sigma_X^2 + \sigma_U^2)D_{\mathrm{d}}}.
                \end{equation}
                
\item[4)]  If \begin{subequations}\label{eq:d2} 
    \begin{equation}
      \label{eq:amos_aldi15}
      \sqrt{D_{\mathrm{e\vphantom{d}}} \sigma_U^2} <
      \min\left\{D_{\mathrm{d}},
        \frac{\sigma_X^2 \sigma_U^2}
        {\sigma_X^2 + \sigma_U^2}
      \right\} 
    \end{equation}
    and 
           \begin{equation}
             \label{eq:amos_aldi20}
             D_{\mathrm{d}}  <
             \sigma_X^2 \left(1 -
               \sqrt{\frac{D_{\mathrm{e}}}{\sigma_U^2}}\right)^2 + D_{\mathrm{e\vphantom{d}}},
                \end{equation}
                \end{subequations}
 then we consider the choice
                \begin{subequations}\label{eq:c2}
                \begin{eqnarray}
  b   & = & \sqrt{\frac{D_{\mathrm{e\vphantom{d}}}}{\sigma_U^2}}, \\
                    a      
  & = & \frac{\sigma_X^2}{\sigma_X^2 + \sigma_W^2} (1 -b),\\
  \sigma_W^2  & = & \frac{\sigma_X^2 (D_{\mathrm{d}} -
  b^2 \sigma_U^2)} {\sigma_X^2 (1 - b)^2 +
  b^2 \sigma_U^2 - D_{\mathrm{d}}} \nonumber \\
  &=&  \frac{\sigma_X^2 (D_{\mathrm{d}} -
  D_\e)} {\sigma_X^2 \left(1 - \sqrt{\frac{D_{\mathrm{e\vphantom{d}}}}{\sigma_U^2}}\right)^2 +
  D_\e - D_{\mathrm{d}}}. \IEEEeqnarraynumspace  \label{eq:sw}
\end{eqnarray}
\end{subequations}

To see that the RHS of \eqref{eq:sw} is positive note
that~\eqref{eq:amos_aldi20} implies that the denominator is positive,
and~\eqref{eq:amos_aldi15} implies that the numerator is positive
because
\begin{multline}
  \label{eq:amos_aldi40}
  \biggl( \sqrt{D_{\mathrm{e\vphantom{d}}} \sigma_U^2} <
      \min\left\{D_{\mathrm{d}},
        \frac{\sigma_X^2 \sigma_U^2}
        {\sigma_X^2 + \sigma_U^2}
      \right\} \biggr) \\ \Longrightarrow \\
    \biggl( D_\e < \min \bigl\{\sigma_{U}^2,D_\d \bigr\} \biggr).
\end{multline}
(Since $\sigma_{X}^{2}/(\sigma_{X}^{2} + \sigma_{U}^{2})$ is
smaller than one, the LHS of \eqref{eq:amos_aldi40} implies that $D_\e
< \sigma_{U}^{2}$. This, and the fact that the LHS
of~\eqref{eq:amos_aldi40} also
implies that $D_{\e} \sigma_{U}^{2} < D_{\d}^{2}$ demonstrates that
the LHS of~\eqref{eq:amos_aldi40}  also implies that $D_\e < D_\d$.)

This choice satisfies~\eqref{block:amos_lunch} because
\begin{align}
  & (1 - a - b)^2 \sigma_X^2 + a^2 \sigma_W^2 + b^2 \sigma_U^2 \nonumber\\
    &  =   \left(\frac{\sigma_W^2 (1 - b)}{\sigma_X^2 +    \sigma_W^2 }\right)^2\sigma_X^2  + \left(\frac{\sigma_X^2    (1 -b)}{\sigma_X^2 + \sigma_W^2}\right)^2 \sigma_W^2 +
    D_{\mathrm{e\vphantom{d}}}\\
    &  =  \frac{\sigma_X^2 (1 -
    b)^2}{\frac{\sigma_X^2}{\sigma_W^2} + 1} +
    D_{\mathrm{e\vphantom{d}}}
    \label{equ:gaussian-ac-dec-dist-our-not-zero}\\
    &  =  \frac{\sigma_X^2 (1 - b)^2 (D_{\mathrm{d}} - b^2 \sigma_U^2)}
    {\sigma_X^2 (1 - b)^2} + D_{\mathrm{e\vphantom{d}}}\\
    &  =  D_{\mathrm{d}}. \label{eq:Dd}
\end{align}
and 
\begin{equation}\label{eq:De}
  b^2  \sigma_U^2
     = D_{\mathrm{e\vphantom{d}}}.
\end{equation}
Moreover, for this choice,
\begin{IEEEeqnarray}{rCl}
\label{eq:Requal}
\lefteqn{
  \frac{1}{2}
    \log
    \frac{\sigma_X^2 \sigma_U^2 + \sigma_X^2 \sigma_W^2 + \sigma_U^2 \sigma_W^2}
         {(\sigma_X^2 + \sigma_U^2)\sigma_W^2}       
    }\qquad \nonumber \\  & = &    \frac{1}{2} \log\frac{\sigma_X^2\big(\sigma_U^2 + D_{\mathrm{d}} - 2 \sqrt{\sigma_U^2
  D_{\mathrm{e\vphantom{d}}}}\big)}{(\sigma_X^2 +
                \sigma_U^2)(D_{\mathrm{d}} -
  D_{\mathrm{e\vphantom{d}}})}.\IEEEeqnarraynumspace
\end{IEEEeqnarray}
Thus, by~\eqref{eq:Dd}--\eqref{eq:Requal} and by Proposition~\ref{prop:Gach}, we conclude that when~\eqref{eq:d2} holds, 
       \begin{equation}\label{equ:gaussian-ac-rate-malaer}
         R^{\textnormal{G}}(D_\d, D_\e) \leq      \frac{1}{2} \log \frac{\sigma_X^2\big(\sigma_U^2 + D_{\mathrm{d}} - 2 \sqrt{\sigma_U^2
  D_{\mathrm{e\vphantom{d}}}}\big)}{(\sigma_X^2 +
                \sigma_U^2)(D_{\mathrm{d}} -
  D_{\mathrm{e\vphantom{d}}})}.
                \end{equation}

\end{enumerate}

\begin{remark}
  The expressions in Proposition~\ref{prop:Gach} and their relation
  to~\eqref{eq:rdf} become more transparent when we define
\begin{subequations}\label{eq:choice}
\begin{IEEEeqnarray}{rcl}
Z&=&a(X+W)\\
\hat{X}_\d &= &bY+Z\\
\hat{X}_\e &=& bX+Z
\end{IEEEeqnarray}
\end{subequations}
for $a>0$, $b\geq0$, and $W$ a centered Gaussian of positive variance
$\sigma_W^2$ independent of the pair $(X,Y)$. With these definitions
\begin{subequations}\label{eq:equivGpar}
\begin{IEEEeqnarray}{rCl}
I(X;Z|Y) &= &  \frac{1}{2}
    \log
    \frac{\sigma_X^2 \sigma_U^2 + \sigma_X^2 \sigma_W^2 + \sigma_U^2 \sigma_W^2}
         {(\sigma_X^2 + \sigma_U^2)\sigma_W^2}       
    \IEEEeqnarraynumspace \\
     \E{(X-\hat{X}_\d)^2} & = & (1-a-b)^2 \sigma_X^2
              + a^2 \sigma_W^2
              + b^2 \sigma_U^2 
              \\
              \E{(\hat{X}_\d- \hat{X}_\e)^2} & = &     b^2 \sigma_U^2.
\end{IEEEeqnarray}
\end{subequations}
Since $Z\markov X\markov Y$ for all choices of the parameters $a>0$,
$b\geq0$, $\sigma_W^2>0$, we can also rewrite~\eqref{eq:Rgup} as:
\begin{equation}
R^{\textnormal{G}}(D_\d, D_\e) \leq \min_{Z, \hat{X}_\d, \hat{X}_\e} I(X;Z|Y)
\end{equation}
where the minimum is over all $Z, \hat{X}_\d, \hat{X}_\e$ that are of the form in~\eqref{eq:choice}
 and satisfy the distortion constraints
\begin{IEEEeqnarray}{rCl}
  \E{\bigl(X - \hat{X}_{\d}\bigr)^2} & \leq & \Dd,\\
  \E{\bigl( \hat{X}_{\d} - \hat{X}_{\e} \bigr)^2} & \leq & \De.
 \end{IEEEeqnarray}
\end{remark}


\subsection{The Converse for Theorem~\ref{thm:gaussian}} \label{sec:convGauss}
If
\begin{equation*}
  \sqrt{D_\e \sigma_U^2} \geq
  \min\Big\{D_\d, \frac{\sigma_X^2 \sigma_U^2}{\sigma_X^2+\sigma_U^2}
  \Big\}
\end{equation*}
then the converse follows by relaxing the constraint \eqref{eq:distE};
see Remark~\ref{rem:wzG}. We thus focus on the case where
\begin{equation}
  \label{eq:cond}
  \sqrt{D_\e \sigma_U^2} < 
  \min\bigg\{D_\d, \frac{\sigma_X^2 \sigma_U^2}{\sigma_X^2+\sigma_U^2} \bigg\}.
\end{equation} 
%


We define the function
$\Ritcntsym\colon \RealsPP \times \RealsP \to \RealsP$ like
$\Ritsym(\cdot,\cdot)$ except that its first argument ($D_\d$) is strictly
positive; the minimum is replaced by an infimum; and the size of the
auxiliary alphabet~$\mathcal{Z}$ can be unbounded.  Thus,
\begin{equation}
  \label{eq:amos_def_RitcntDD}
  \RitcntDD \triangleq \inf_{Z, \phi, \psi} I(X;Z|Y)
\end{equation}
where the infimum is over all choices\footnote{To be more precise we
  should specify the set where $Z$ may take value, and we must
  restrict the functions $\phi$ and $\psi$ to be measurable. In the
  converse~$Z$ will correspond to the tuple $(M, Y^{i-1},
  Y_{i+1}^{n})$, and we can therefore restrict $Z$ here to be the
  space where such tuples take value.}  of the random variable~$Z$ and
functions $\phi, \psi$ satisfying
\begin{subequations}
\label{eq:amos600}
\begin{IEEEeqnarray}{rCl}
  \E{(X - \hat{X}_{\d})^2} & \leq & \Dd,
    \label{equ:gaussian-cv-max-h-x-dec-constraint}\\
  \E{(\hat{X}_{\d} - \hat{X}_{\e})^2} & \leq & \De,
    \label{equ:gaussian-cv-max-h-x-enc-constraint}\\
  Z \markov& X& \markov Y, \label{eq:markov} 
  \end{IEEEeqnarray} 
  where
  \begin{IEEEeqnarray}{rCl}
  \hat{X}_\d & \triangleq &\phi(Y,Z),\label{eq:phi}\\
  \hat{X}_\e & \triangleq & \psi(X,Z).\label{eq:psi}
\end{IEEEeqnarray}
\end{subequations}

In analogy to Proposition~\ref{prop:key_properties} we have:
\begin{lemma}
  \label{lem:Gaus_mon_convex}
  Over $\RealsPP \times \RealsP$ the function $\RitcntDD$ is finite;
  monotonic in each of its arguments; and convex.
\end{lemma}
\begin{proof}
  The function is bounded by the rate-distortion function of the Gaussian
  source without side information. The proof of monotonicity is
  identical to the proof of monotonicity in
  Proposition~\ref{prop:key_properties}. The proof of convexity is
  also very similar; only a minor change is needed to account for the
  fact that, \emph{prima facie}, the infimum need not be achieved.
\end{proof}

The following lemma provides an explicit expression for $\RitcntDD$
when~\eqref{eq:cond} holds.
\begin{lemma}
  \label{lemma:amos_aldi100}
  If $D_\d > 0$ and $D_\e \geq 0$ satisfy~\eqref{eq:cond}, then
\begin{equation}\label{equ:gaussian-cv-lb-2show-malaer}
\RitcntDD=    \frac{1}{2} \log^+\left(\frac{\sigma_X^2}{\sigma_X^2 + \sigma_U^2}
  \frac{\sigma_U^2 + D_{\mathrm{d}} - 2 \sqrt{\sigma_U^2
  D_{\mathrm{e\vphantom{d}}}}}{D_{\mathrm{d}} -
  D_{\mathrm{e\vphantom{d}}}}\right).
\end{equation} 
\end{lemma}
\begin{proof}[Proof of Lemma~\ref{lemma:amos_aldi100}]
We first prove
\begin{equation} \label{eq:Rlow}
\RitcntDD \leq    \frac{1}{2} \log^+\left(\frac{\sigma_X^2}{\sigma_X^2 + \sigma_U^2}
  \frac{\sigma_U^2 + D_{\mathrm{d}} - 2 \sqrt{\sigma_U^2
  D_{\mathrm{e\vphantom{d}}}}}{D_{\mathrm{d}} -
  D_{\mathrm{e\vphantom{d}}}}\right).
\end{equation}
To this end, we present a choice for $Z$, $\hat{X}_\d$, $\hat{X}_\e$
that satisfies the constraints~\eqref{eq:amos600} and is such that the
objective function $I(X;Z|Y)$ in~\eqref{eq:amos_def_RitcntDD}
evaluates to the RHS of~\eqref{eq:Rlow}. Our choice depends on whether
\begin{equation}\label{eq:DdDe1}
  D_{\mathrm{d}} \geq 
  \sigma_X^2 
  \left(1 -
    \sqrt{\frac{D_{\mathrm{e}}}{\sigma_U^2}}\right)^2 
  + D_{\mathrm{e\vphantom{d}}}
\end{equation}
or 
\begin{equation}\label{eq:DdDe2}
  D_{\mathrm{d}} < 
  \sigma_X^2 \left(1 -
    \sqrt{\frac{D_{\mathrm{e}}}{\sigma_U^2}}\right)^2 + 
  D_{\mathrm{e\vphantom{d}}}.
\end{equation}  
In the first case~\eqref{eq:DdDe1} the RHS of~\eqref{eq:Rlow}
evaluates to 0, whereas in the second case~\eqref{eq:DdDe2} it is
positive.

When $D_\d$ and $D_\e$ satisfy~\eqref{eq:DdDe1}, a suitable
choice is---as in \eqref{eq:Xe1} and \eqref{eq:Xd1} in the proof of the
direct part---
\begin{equation}\label{eq:cc1}
  Z=\emptyset, \qquad \hat{X}_\e^n =  \sqrt{\frac{D_{\mathrm{e\vphantom{d}}}}{\sigma_U^2}} X^n, \qquad \hat{X}_\d^n =  \sqrt{\frac{D_{\mathrm{e\vphantom{d}}}}{\sigma_U^2}} Y^n.
\end{equation}

When $D_\d$ and $D_\e$ satisfy~\eqref{eq:DdDe2}, a suitable
choice is---as in~\eqref{eq:c2} and \eqref{eq:choice} in the
direct part---
\begin{IEEEeqnarray}{rCl}
  Z&=&a(X+W),\quad  \hat{X}_\e=bX+Z, \quad \hat{X}_\d=bY+Z,\IEEEeqnarraynumspace
\end{IEEEeqnarray}
where $W$ is a centered Gaussian of variance $\sigma_W^2=
\frac{\sigma_X^2 (D_{\mathrm{d}} - D_\e)} {\sigma_X^2 (1 -
  \sqrt{D_{\mathrm{e\vphantom{d}}}/\sigma_U^2})^2 + D_\e -
  D_{\mathrm{d}}}$ and independent of the pair $(X,Y)$ and where $b =
\sqrt{ D_{\mathrm{e\vphantom{d}}}/\sigma_U^2}$ and $a =
\frac{\sigma_X^2}{\sigma_X^2 + \sigma_W^2} (1 -b)$. That this choice
has the desired properties follows by~\eqref{eq:Dd}--\eqref{eq:Requal}
and~\eqref{eq:equivGpar}.
               
Having established~\eqref{eq:Rlow}, we now complete the proof of the
lemma by proving the reverse inequality
\begin{equation}
\RitcntDD \geq     \frac{1}{2} \log^+\left(\frac{\sigma_X^2}{\sigma_X^2 + \sigma_U^2}
  \frac{\sigma_U^2 + D_{\mathrm{d}} - 2 \sqrt{\sigma_U^2
  D_{\mathrm{e\vphantom{d}}}}}{D_{\mathrm{d}} -
  D_{\mathrm{e\vphantom{d}}}}\right).
\end{equation}
Since rates are nonnegative, it suffices to prove
\begin{equation}\label{eq:convGa}
R^{\textnormal{G}}(D_\d, D_\e) \geq
  \frac{1}{2} \log\left(\frac{\sigma_X^2}{\sigma_X^2 + \sigma_U^2}
  \frac{\sigma_U^2 + D_{\mathrm{d}} - 2 \sqrt{\sigma_U^2
  D_{\mathrm{e\vphantom{d}}}}}{D_{\mathrm{d}} -
  D_{\mathrm{e\vphantom{d}}}}\right)
\end{equation}
where $\logplus$ has been replaced by $\log$.

Since the joint law of $(X,Y)$ is fixed and is a bivariate Gaussian law
\begin{IEEEeqnarray}{rCl}
  I(X;Z|Y)
& = & h(X|Y) - h(X|Y,Z) \nonumber \\
           & = & \frac{1}{2}
                 \log\left(2 \pi e
                 \frac{\sigma_X^2 \sigma_U^2}
                      {\sigma_X^2 + \sigma_U^2}
                     \right)
                 - h(X|Y,Z).\IEEEeqnarraynumspace
                 \label{equ:gaussian-cv-i-as-diff-log}
\end{IEEEeqnarray}
Consequently, \eqref{eq:convGa} is equivalent to 
\begin{equation}
  \label{eq:up}
  \Omega \leq \frac{1}{2} \log\left( 2\pi e \sigma_U^2 
  \frac{D_{\mathrm{d}} -
  D_{\mathrm{e\vphantom{d}}}}{\sigma_U^2 + D_{\mathrm{d}} - 2 \sqrt{\sigma_U^2
  D_{\mathrm{e\vphantom{d}}}}}\right),
\end{equation}
where $\Omega$ is defined as
\begin{equation}
  \label{eq:omega}
  \Omega \triangleq \sup_{Z, \phi, \psi} h(X|Y,Z)
\end{equation}
under the same constraints~\eqref{eq:amos600} that define $\RitcntDD$ in~\eqref{eq:amos_def_RitcntDD}.
 
To prove~\eqref{eq:up} we first note that, since $\hat{X}_\d$ is a
deterministic function of $(Y,Z)$,
\begin{IEEEeqnarray}{rCl}
  h(X|Y,Z) & = & h(X - \hat{X}_{\mathrm{d}}|Y,Z, \hat{X}_\d)\\
    & = & h(X - \hat{X}_{\mathrm{d}}|X - \hat{X}_{\mathrm{d}} + U,
                                     Z,\hat{X}_{\mathrm{d}})\\
    & \leq & h(X - \hat{X}_{\mathrm{d}}|X - \hat{X}_{\mathrm{d}} + U)
             \label{equ:gaussian-cv-h-x-upper-bounded-by-h-xxd-cond}
             \label{equ:gaussian-cv-h-x-upper-bounded-by-log-var}
\end{IEEEeqnarray}
where in the second line we recalled that $Y = X + U$~\eqref{eq:Y}, and
where the last line follows because conditioning cannot increase
differential entropy.

The Markov condition $Z \markov X \markov Y$ \eqref{eq:markov} and the
fact that $Y = X + U$~\eqref{eq:Y} imply that
\begin{equation}
  Z \markov X \markov U.
\end{equation}
This, combined with the assumption that $U$ is independent of~$X$,
implies that $U$ is independent of $(X,Z)$. And since $\hat{X}_\e$ is
a function of $(X,Z)$, 
\begin{equation}
  \text{$U$ and $(\hat{X}_\e,X,Z)$ are independent.}
\end{equation}
This independence implies that $U$ is independent of $(X -
\hat{X}_\e)$. This latter independence and the fact that
$X-\hat{X}_\d$ can be expressed as $-\bigl( \hat{X}_\d - \hat{X}_\e - (
X - \hat{X}_\e) \bigr)$ implies that
\begin{equation}\label{eq:conv}
   \operatorname{Cov}(X-\hat{X}_\d,U)=- \operatorname{Cov}(\hat{X}_\d-\hat{X}_\e,U). 
\end{equation}
From \eqref{eq:conv}, \eqref{equ:gaussian-cv-max-h-x-enc-constraint},
the fact that
the variance of a random variable cannot exceed its second moment, and
the fact that the magnitude of a correlation coefficient cannot exceed
$1$, it follows that
\begin{eqnarray}\label{eq:conddesu}
|\operatorname{Cov}({X}-\hat{X}_\d, U)|^2&\leq &\De \, \sigma_U^2.
\end{eqnarray}

From~\eqref{equ:gaussian-cv-h-x-upper-bounded-by-log-var} and
\eqref{eq:conddesu} we thus obtain
\begin{equation}\label{eq:GO}
\Omega  \leq   \Gamma
     \end{equation}
     where $\Gamma$ is defined as 
\begin{equation}\label{eq:gammap}
  \Gamma \triangleq
  \sup_{\hat{X}_\d} h(X-\hat{X}_\d|X -\hat{X}_{\d} + U)
\end{equation}
subject to the relaxed constraints
\begin{subequations}
  \label{eq:amos650}
  \begin{eqnarray}
    \operatorname{Var}(X - \hat{X}_{\d})\label{eq:gameq1}
    & \leq & \Dd,\\
    \big|\operatorname{Cov}(X - \hat{X}_{\d}, U)\big|^2
    & \leq & \De \, \sigma_U^2.\label{eq:gameq2}
  \end{eqnarray}
\end{subequations}

We now proceed to study $\Gamma$. Define
\begin{equation}
  A \triangleq X-\hat{X}_\d
\end{equation}
so
\begin{equation}\label{eq:A_gammap}
  \Gamma  = 
  \sup_{A} h(A | A +  U)
\end{equation}
subject to 
\begin{subequations}
  \label{eq:Aamos650}
  \begin{eqnarray}
    \operatorname{Var}(A)\label{eq:Agameq1}
    & \leq & \Dd,\\
    \big|\operatorname{Cov}(A, U)\big|^2
    & \leq & \De \, \sigma_U^2.\label{eq:Agameq2}
  \end{eqnarray}
\end{subequations}

By the conditional max-entropy theorem~\cite{thomas87}, the supremum
in \eqref{eq:A_gammap} is achieved when $(A,U)$ are jointly Gaussian,
as we henceforth assume.
As we next argue, the lemma's hypothesis that~\eqref{eq:cond} holds
implies that the choice of $A$ as $-U$ is not in the feasible
set. Indeed, with this choice $|\operatorname{Cov}(A, U)|^2$ is equal
to~$\sigma_U^4$, which violates \eqref{eq:Agameq2}
because~\eqref{eq:cond} and~\eqref{eq:amos_aldi40} imply
  \begin{equation}\label{eq:cond2}
 D_\e < \min\{\sigma_{U}^2,D_\d\}.
 \end{equation}
We thus assume in the following that $A$ is jointly Gaussian with $U$
and that $A\neq -U$. Consequently,
 \begin{IEEEeqnarray}{rCl}
\lefteqn{h(A|A+U)}\qquad \nonumber \\& =&
 \frac{1}{2} \log \left(2 \pi e\left( \sigma_U^2 - \frac{( \sigma_{U}^2 +
                \kappa_{AU})^2}
               {\sigma_{A}^2 + \sigma_{U}^2 + 2\kappa_{AU}}\right) \right)   \label{eq:diffentropy} \\
              &=& \frac{1}{2} \log \left(2 \pi e \frac{\sigma_A^2 \sigma_{U}^2 -
                \kappa_{AU}^2}
               {\sigma_{A}^2 + \sigma_{U}^2 + 2\kappa_{AU}} \right)
\end{IEEEeqnarray}
 where $\sigma_A^2\triangleq \operatorname{Var}(A)$ and $\kappa_{AU}\triangleq
\operatorname{Cov}(A,U)$. 

We can thus rewrite the optimization problem in \eqref{eq:gammap} as
 \begin{equation}\label{equ:gaussian-cv-max-h-a}
  \Gamma = \sup_{\kappa_{AU}, \sigma_A^2} \frac{1}{2} \log \left(2 \pi e \frac{\sigma_A^2 \sigma_{U}^2 -
                \kappa_{AU}^2}
               {\sigma_{A}^2 + \sigma_{U}^2 + 2\kappa_{AU}} \right)
\end{equation}
subject to 
\begin{eqnarray}\label{eq:var}
 0\leq  &\sigma_A^2   & \leq 
  \Dd,\\ \label{eq:cov}\label{eq:cond3}
0\leq & |\kappa_{AU}|^2 & \leq 
  \De  \sigma_U^2, \\
  0\leq & |\kappa_{AU}|^2 & \leq  \sigma_A^2\sigma_U^2.
  \label{equ:gaussian-cv-max-h-a-cov-constraint}
\end{eqnarray}
(We have to add the last constraint because the magnitude of a
correlation coefficient cannot exceed one.)  For fixed $\kappa_{AU}$,
the objective function in \eqref{equ:gaussian-cv-max-h-a} is
monotonically increasing in $\sigma_{A}^2$ (see
also~\eqref{eq:diffentropy}), and so is the RHS of
Constraint~\eqref{equ:gaussian-cv-max-h-a-cov-constraint}. Therefore,
it is optimal to choose in~\eqref{equ:gaussian-cv-max-h-a}
\begin{equation}\label{eq:opts}
\sigma_{A}^2 = \Dd.
\end{equation}
Substituting this choice in \eqref{equ:gaussian-cv-max-h-a} and
\eqref{equ:gaussian-cv-max-h-a-cov-constraint} yields
\begin{equation}\label{eq:opt2}
\Gamma = \sup_{\kappa_{AU}}  \frac{1}{2} \log \left(2 \pi e \frac{D_\d \sigma_{U}^2 -
                \kappa_{AU}^2}
               {D_\d + \sigma_{U}^2 + 2\kappa_{AU}} \right)
\end{equation}
subject to \eqref{eq:cov} and 
\begin{equation}
0\leq  |\kappa_{AU}|^2  \leq 
  \Dd \, \sigma_U^2 .
  \label{eq:cond3a}
\end{equation}
Notice that, whenever~\eqref{eq:cond} holds, the RHS
of~\eqref{eq:cond3} is upper-bounded by the square of
$\min\{D_\d,\sigma_U^2\}$. Consequently,
\begin{equation}\label{eq:cc}
\Bigl( \textnormal{\eqref{eq:cond} and \eqref{eq:cond3}} \Bigr)
\Rightarrow
\Bigl( |\kappa_{AU}| < \min\{D_\d,\sigma_U^2\} \Bigr).
\end{equation} 
Since the RHS of~\eqref{eq:cc} implies \eqref{eq:cond3a}, 
\begin{equation}
  \Bigl( \textnormal{\eqref{eq:cond} and \eqref{eq:cond3}} \Bigr)
  \Rightarrow \textnormal{\eqref{eq:cond3a}},
\end{equation}
and Constraint~\eqref{eq:cond3a} is redundant.  We therefore ignore
Constraint~\eqref{eq:cond3a} and study the maximization
in~\eqref{eq:opt2} subject to \eqref{eq:cov} only.

To this end, we compute the derivative of the objective function in
\eqref{eq:opt2} with respect to $\kappa_{AU}$:
\begin{IEEEeqnarray}{rCl}
\lefteqn{\frac{ \d }{\d \kappa_{AU}} \left(  \frac{1}{2} \log \left(2 \pi e \frac{D_\d \sigma_{U}^2 -
                \kappa_{AU}^2}
               {D_\d + \sigma_{U}^2 + 2\kappa_{AU}} \right)\right)}\qquad \nonumber \\
               & = & 
    \frac{- (D_\d+\kappa_{AU})(\sigma_U^2+\kappa_{AU}) }{(\Dd + \sigma_{U}^2 + 2 \kappa_{AU})(D_\d\sigma_U^2-\kappa_{AU}^2)}.\label{eq:dev}
\end{IEEEeqnarray}
By~\eqref{eq:cc}, the derivative in \eqref{eq:dev} is negative for all
feasible $\kappa_{AU}$. Hence, the objective function in
\eqref{eq:opt2} is decreasing on the (symmetric) interval of
interest~\eqref{eq:cond3}, and it is optimal to choose
\begin{equation}\label{eq:optk}
\kappa_{AU}= - \sqrt{D_\e \sigma_U^2}.
\end{equation}
The optimality of this choice allows us to evaluate $\Gamma$ via
\eqref{eq:opt2} and hence to upper-bound $\Omega$
via~\eqref{eq:GO}. This yields the desired bound~\eqref{eq:up}, which
establishes the lemma.
 \end{proof}


\begin{proof}[Proof of Converse when \eqref{eq:cond} holds]
  Using Lemma~\ref{lem:Gaus_mon_convex} and Lemma~\ref{lemma:amos_aldi100}
  we can follow the steps of the proof in Section~\ref{sec:th1} of the
  converse part of Theorem~\ref{thm:finite}. The remaining
  technicality is continuity. Continuity in the interior, i.e., on
  $\RealsPP \times \RealsPP$ follows from convexity. It thus only
  remains to establish continuity when $D_{\d} > 0$, \eqref{eq:cond}
  holds, and $D_{\e}$ is zero. This can be done by
  inspecting~\eqref{equ:gaussian-cv-lb-2show-malaer}.
\end{proof}



\section{More and More-General Constraints}\label{sec:Kdistortions}

So far we have only studied settings with two distortion functions,
one of which---the decoder-side distortion function $d_{\d}(x,
\hat{x}_\d)$---depends on the source symbol and the decoder's
reconstruction, and the other---the encoder-side distortion function
$d_{\e}(\hat{x}_{\d}, \hat{x}_\e)$---depends on the decoder's and the
encoder's reconstruction symbols. In this section we extend our
setting to allow for more than two distortion functions and to allow
for distortions that depend on all three symbols: the source symbol
$x$, the decoder's reconstruction symbol $\hat{x}_{\d}$, and the
encoder's reconstruction symbol $\hat{x}_{\e}$ . We shall also allow
the reconstruction alphabets to differ. But all alphabets are assumed
finite.


\subsection{Problem Statement}\label{sec:problem_extended}
The new setup differs from the setup in Section~\ref{sec:setup} in two
ways.
\begin{itemize}
\item The encoder-side reconstruction $\hat{X}_\e^n$ and the
  decoder-side reconstruction $\hat{X}_\d^n$ take value in the finite
  alphabets $\hat{\mathcal X}_\e^n$ and $\hat{\mathcal X}_\d^n$ which
  can be different.
\item There are $K$ (possibly larger than $2$) distortion constraints
  specified by the $K$ distortion functions
\begin{equation}\label{eq:distme}
d_k \colon \mathcal{X} \times \mathcal{X}_\d\times  \mathcal{X}_\e \to \RealsP, \quad k\in\{1,\ldots, K\}
\end{equation} 
and the corresponding $K$ maximal-allowed distortions $D_1,\ldots,
D_K$ (all of which are assumed to be nonnegative).
\end{itemize}
We say that the tuple $(R, D_1, \ldots, D_K)$ is \emph{achievable} if
for every $\epsilon>0$ and sufficiently large $n$ there exist a
message set $\set{M}$ of size $|\set{M}| \leq 2^{n(R+\epsilon)}$ 
and functions
\begin{subequations}\label{eq:funK}
\begin{IEEEeqnarray}{rCl}
f^{(n)}& \colon& \set{X}^{n} \to \set{M}\\
\phi^{(n)} & \colon & \set{M} \times \set{Y}^n \to \hat{\set{X}}_{\d}^n\\
\psi^{(n)} & \colon & \set{X}^n \to \hat{\set{X}}^n_\e
\end{IEEEeqnarray}
\end{subequations}
such that  the message $M=\encn(X^n)$ and the reconstruction
sequences $\hat{X}_\d^n=\decn(M,Y^n)$ and
$\hat{X}_\e^n=\ercn(X^n)$ satisfy:
\begin{IEEEeqnarray}{rCl} \label{eq:distck}
\frac{1}{n} \sum_{i=1}^n \E{d_{k}(X_i, \hat{X}_{\d,i}, \hat{X}_{\e,i}) } \leq D_{k}+\epsilon, \quad k\in\{1,\ldots, K\}.\nonumber \\
\end{IEEEeqnarray}

In analogy to Assumption~\ref{assumption}, we shall assume:
\begin{assumption}\label{ass:extended}
To each $x\in \set{X}$ corresponds some $\hat{x}_\d\in\hat{X}_\d$ and
some $\hat{x}_\e\in \hat{X}_\e$ satisfying
\begin{equation}\label{eq:assumptionK}
d_k(x, \hat{x}_\d, \hat{x}_\e)=0, \quad k\in\{1,\ldots, K\}.
\end{equation} 
\end{assumption}

We seek the smallest rate $R$ for which the tuple $(R, D_1,\ldots,
D_K)$ is achievable. This is defined as follows.  Given
a maximal-allowed-distortion tuple $(D_1,\ldots, D_K)$, let
\begin{eqnarray}\lefteqn{
\set{R}_{\textnormal{Ext}}(D_1,\ldots, D_K)} \qquad\nonumber  \\
&\triangleq& \{R \in \RealsP\colon (R, D_1,\ldots, D_K) \textnormal{ is achievable}\}. \IEEEeqnarraynumspace
\end{eqnarray}
Assumption~\ref{ass:extended} implies that the set
$\set{R}_{\textnormal{Ext}}(D_1, \ldots, D_K)$ contains all rates
exceeding $H(X|Y)$ and is thus nonempty.  The
rate-distortions function $R_{\textnormal{Ext}}$ can now be defined as
\begin{equation}
R_{\textnormal{Ext}}(D_1, \ldots, D_K) \triangleq \min_{R\in \rateregion_{\textnormal{Ext}}(D_1,\ldots, D_K)}{R},
\end{equation}
where the  minimum exists because the region  $\set{R}_{\textnormal{Ext}}(D_1,\ldots, D_K) \subset \RealsP$ is nonempty, closed, and bounded from below by 0.

\subsection{Result}

To describe the rate-distortions function for the extended setup of
Section~\ref{sec:problem_extended}, we next introduce the function
$\tilde{R}_{\textnormal{Ext}}(D_1,\ldots, D_K)$.

Given the joint law $P_{XY}$ of the source and side information, and
given the distortion functions $d_1, \ldots, d_K$, this function is
defined as
\begin{equation}
\tilde{R}_{\textnormal{Ext}}(D_1,\ldots, D_K) =\min_{ U,Z, \fdc, \fec} \bigl( I(X;Z)-I(Y;Z) \bigr)
\end{equation}
where the minimization is over 
all discrete auxiliary random variables $Z$ and $U$ satisfying
\begin{equation}
\label{eq:amos4001}
(U,Z)\markov X\markov Y
\end{equation}
and over all functions $\fdc\colon \mathcal{Y} \times \mathcal{Z} \to
\hat{\mathcal{X}}_\d$ and $\fec\colon \mathcal{X} \times
\mathcal{Z}\times \mathcal{U} \to \hat{\mathcal{X}}_\e$ that
simultaneously satisfy the $K$ distortion constraints
\begin{IEEEeqnarray}{rCl}
\E{ d_k \bigl( X, \fdc(Y,Z), \psi(X,Z,U) \bigr)} \leq D_k, 
\quad k\in\{1,\ldots, K\}. \nonumber\\
\label{eq:distErk}
\end{IEEEeqnarray}

The following proposition provides cardinality bounds on the support
sets of the auxiliary random variables.
\begin{proposition}[Cardinality Bounds]
  \label{prop:AmosCardBound}
  The minimum defining $\tilde{R}_{\textnormal{Ext}}(D_1,\ldots, D_K)$
  is not increased if we restrict the cardinality of the support
  set~$\mathcal{Z}$ of $Z$ to
  \begin{equation}
    |\mathcal{Z}| \leq |\mathcal{X}||\set{U}|+K+1
  \end{equation}
  and the cardinality of the support set $\set{U}$ of $U$ to
  \begin{equation}
    \label{eq:amos_cardU}
    |\mathcal{U}| \leq K.
  \end{equation}
\end{proposition}
\begin{proof}
  The cardinality bound on $\mathcal{Z}$ can be justified using the
  convex cover method \cite{ElGamalKim2011}. The cardinality bound on
  $\set{U}$ is proved in Appendix~\ref{app:Amos}.
\end{proof}

\begin{remark}[Improved Cardinality Bound]
  The cardinality bound on $\set{U}$ can be strengthened: $|\set{U}|$
  need not exceed the number of distortion constraints
  in~\eqref{eq:distck} that depend on $\hat{X}_{\e,i}$.  The latter
  number equals $1$ in the original setup of Section~\ref{sec:setup}
  thus allowing us to recover Theorem~\ref{thm:finite}.
\end{remark}

\begin{proposition}[Key Properties of the Function $\tilde{R}_{\textnormal{Ext}}$]
  \label{prop:key_propertiesK}
  The function $\tilde{R}_{\textnormal{Ext}}\colon \RealsP^{K} \to
  \RealsP$ is bounded from above by $H(X|Y)$; it is nondecreasing in
  the distortions
\begin{multline*}
\Bigl( D_1'\geq D_1,\; \ldots,  \; D_K' \geq D_K \Bigr) \\
\Longrightarrow \quad 
 \Bigl( \tilde{R}_{\textnormal{Ext}}(D_1', \ldots, D_K') \leq \tilde{R}_{\textnormal{Ext}}(D_1, \ldots, D_K) \Bigr);
\end{multline*}
and it is convex and continuous.
\end{proposition}
\begin{proof} 
  The proof is similar to the proof of
  Proposition~\ref{prop:key_properties} in
  Appendix~\ref{app:key_properties} and is omitted.
\end{proof}

\begin{theorem}\label{thm:extension}
  The rate-distortions function for the setup in
  Section~\ref{sec:problem_extended} is equal to
  $\tilde{R}_{\textnormal{Ext}}(D_1, \ldots, D_K)$:
\begin{equation}
  \label{eq:rdf1}
  R_{\textnormal{Ext}}(D_1, \ldots, D_K) = 
  \tilde{R}_{\textnormal{Ext}}(D_1, \ldots, D_K).
\end{equation}
\end{theorem}
\begin{proof}
 The achievability, i.e., that 
 \begin{equation}
 R_{\textnormal{Ext}}(D_1, \ldots, D_K)\leq \tilde{R}_{\textnormal{Ext}}(D_1, \ldots, D_K),
\end{equation} 
can be proved using a scheme that is similar to the one that was
sketched in the proof of Theorem~\ref{thm:finite}. The only difference
is that, to produce the reconstruction sequence $\hat{X}_\e^n$, the
encoder applies the function~$\psi$ component-wise to the tuple $(X^n,
Z^{*n}, U^n)$, where, conditional on $(X^n, Z^{*n})$, the components
of the sequence $U^n$ are generated independently according to the
conditional law $P_{U|Z,X}$. The analysis of this scheme is omitted.
 
 
 We next prove the converse, i.e., that
 \begin{equation}\label{eq:ext_converse}
   R_{\textnormal{Ext}}(D_1, \ldots, D_K)\geq 
   \tilde{R}_{\textnormal{Ext}}(D_1, \ldots, D_K).
 \end{equation}
 
 Fix some positive $\epsilon$, a blocklength $n$, and a rate $R$. Let
 $\set{M}$ be a message set of size $|\set{M}|\leq 2^{n(R+\epsilon)}$,
 and let $\encn$, $\phi^{(n)}$, and $\psi^{(n)}$ be encoding and
 reconstruction functions as in \eqref{eq:funK} that satisfy the $K$
 distortion constraints in \eqref{eq:distck}.  For every
 $i\in\{1,\ldots, n\}$, define $Z_i$ in
 \eqref{eq:Zi} 
 \begin{equation} Z_i \triangleq (\idx,Y^{i-1},Y_{i+1}^n)
 \end{equation}
 and  define $U_i$ as 
 \begin{equation} U_i\triangleq (X_1^{i-1}, X_{i+1}^n).
 \end{equation}
 Notice that  for every $i\in\{1,\ldots, n\}$
 \begin{equation}
   (U_i, Z_i ) \markov X_i \markov Y_i.
 \end{equation}
 Also, following the steps
 in~\eqref{eq:25}--\eqref{equ:cv-part1-last}, we can conclude that
 \begin{equation}\label{eq:rI}
   n(R+\epsilon) \geq \sum_{i=1}^n I(X_i;Z_i)- I(Y_i;Z_i).
 \end{equation}
 
We further define---as in Section~\ref{sec:th1}---$\phi_i^{(n)}$ to be the function that maps $(M,Y^n)$ to the $i$-th symbol of  $\phi^{(n)}(M,Y^n)$ and 
$\psi_i^{(n)}$ to be the function that maps $X^n$ to the $i$-th symbol of  $\psi^{(n)}(X^n)$.
             Then, the symbol $\phi_i^{(n)}(M,Y^n)$ can be written as
\begin{equation}\label{eq:def0}
  \deci(Y_i,Z_i) \triangleq \decni(\idx,Y^n),
\end{equation}
and $\psi_i^{(n)}(X^n)$  can be written as
\begin{equation}\label{eq:def00}
\psi_i(X_i,Z_i,U_i) \triangleq \psi_{i}^{(n)}(X^n),
\end{equation}
for some functions $\deci$ and $\psi_i$ with arguments in the respective domains.
We finally define for each $k\in\{1,\ldots, K\}$ and $i\in\{1,\ldots, n\}$
\begin{equation}\label{eq:def20}
  D_{k,i} \triangleq
  \E{d_{k}(X_i, \decni(\idx,Y^n), \psi_{i}^{(n)}(X^n))},
\end{equation}
where $\Exp[\cdot]$ is with respect to
$P_{X_{\vphantom{I}}^nY_{\vphantom{I}}^n}$. Notice that 
\begin{equation}\label{eq:Dktot}
\sum_{i=1}^n D_{k,i} \leq D_k +\epsilon, \qquad k\in\{1, \ldots, K\}
\end{equation}
because the chosen encoding and reconstruction functions 
$\encn$, $\phi^{(n)}$, and $\psi^{(n)}$ satisfy~\eqref{eq:distck}.
Moreover, by definitions \eqref{eq:def0}--\eqref{eq:def20}, 
\begin{equation}\label{eq:Dsl}
  \E{d_k \bigl( X_i,\deci(Y_i,Z_i), \psi_i(X_i, Z_i, U_i) \bigr)} = D_{k,i},
\end{equation}
where $\Exp[\cdot]$ is with respect to
$P_{X_iY_i\vphantom{|}}P_{U_iZ_i|X_i}$.

Combining~\eqref{eq:rI} and \eqref{eq:Dsl} with the definition of
$\tilde{R}_{\textnormal{Ext}}$, we
obtain 
\begin{IEEEeqnarray}{rCl}
 n(R+\epsilon) 
& {\geq} &  \sum_{i = 1}^n I(X_i;Z_i) - I(Y_i;Z_i)\\
  \quad\quad\quad & {\geq} &
    \sum_{i = 1}^n \tilde{R}_{\textnormal{Ext}}(D_{1,i}, \ldots, D_{K,i}) \label{eq:Rext1}
  \\
  & \geq & 
    n \tilde{R}_{\textnormal{Ext}}\bigg(\frac{1}{n} \sum_{i = 1}^n D_{1,i}, \ldots,
        \frac{1}{n} \sum_{i = 1}^n D_{K,i}\bigg) \\
  & {\geq} & n \tilde{R}_{\textnormal{Ext}}\big( D_1+\epsilon, \ldots, D_K+\epsilon),\label{eq:fina}
\end{IEEEeqnarray}
where the last two inequalities follow by the convexity and the
monotonicity of $\tilde{R}_{\textnormal{Ext}}$ and
by~\eqref{eq:Dktot}.  By the continuity of
$\tilde{R}_{\textnormal{Ext}}$ and because $\epsilon>0$ and the
blocklength $n$ are arbitrary, the converse~\eqref{eq:ext_converse}
follows immediately from~\eqref{eq:fina}.
\end{proof}


 \section*{Acknowledgment}
 We acknowledge helpful discussions with Prof.~G.~Kramer. 

 \appendices

\section{Proof of Corollary~\ref{cor:steinberg}}\label{app:RDcr}
When $d_\e(\cdot,\cdot)$ is the Hamming distortion and $D_\e=0$, our
average-per-symbol distortion constraint \eqref{eq:distE} is less
stringent than the block-distortion constraint~\eqref{eq:steinberg2}
in Steinberg's setup (Remark~\ref{rem:steinberg}). Consequently,
\begin{equation}\label{eq:dir1}
R_{\textnormal{cr}}(D_\d) \geq R(D_\d,0).
\end{equation}
It remains to prove the reverse inequality. Let $Z$, $\phi$, and $\psi$
be minimizers of $R(D_\d,0)$, so
\begin{subequations}
\begin{equation}
  \label{eq:amos_ayef10}
R(D_\d,0) = I(X;Z)-I(Y;Z)
\end{equation}
\begin{equation}
  \label{eq:amos_ayef15}
  \Exp\big[d_\d\bigl(X, \fdc( Y,Z)\bigr)\big]  \leq  \Dd
\end{equation}
\begin{equation}
  \label{eq:amos_ayef16}
  \phi(Y,Z)=\psi(X,Z) \quad \textnormal{w.p.~1}  
\end{equation}
\begin{equation}
  \label{eq:amos_shikor10}
  Z \markov X \markov Y.
\end{equation}
\end{subequations}
To prove the reverse inequality we shall {upper-bound}
$R_{\textnormal{cr}}(D_\d)$ by showing that 
\begin{equation}
  \label{eq:amos_def_X_hat}
  \hat{X}\triangleq \phi(Y,Z)
\end{equation}
is feasible in the minimization \eqref{eq:amos_ayef200} that defines it.

From the definition of $\hat{X}$ \eqref{eq:amos_def_X_hat} and
from~\eqref{eq:amos_ayef16}, it follows that $\hat{X}$ is computable
(almost surely) from $(X,Z)$. This combines
with~\eqref{eq:amos_shikor10} to establish that
\begin{equation}
  \label{eq:MC1} 
  (\hat{X}, Z) \markov X \markov Y
\end{equation}  
and, \emph{a fortiori}, that
\begin{subequations}
\label{block:amos_ayef30}
\begin{equation}
  \label{eq:amos_ayef20}
  \hat{X} \markov X \markov Y.
\end{equation}
And by~\eqref{eq:amos_ayef15} and~\eqref{eq:amos_def_X_hat},
\begin{equation}
  \Exp\big[d_\d\bigl(X, \hat{X}\bigr)\big]  \leq  \Dd.
\end{equation}
\end{subequations}
It follows from~\eqref{block:amos_ayef30} that $\hat{X}$ is feasible in the
minimization \eqref{eq:amos_ayef200} defining
$R_{\textnormal{cr}}(D_\d)$ and thus
\begin{IEEEeqnarray}{rCl}
R_{\textnormal{cr}}(D_\d) & \leq & I(X;\hat{X})- I(Y;\hat{X}) \\
& =  & I(X;\hat{X}|Y) \label{eq:amos_ayef30b} \\
& \leq  & I(X;Z|Y) \label{eq:amos_ayef30c} \\
& =& I(X;Y)-I(X;Z) \label{eq:amos_ayef30d} \\
& = & R(D_\d,0) \label{eq:amos_ayef30e}
\end{IEEEeqnarray}
where~\eqref{eq:amos_ayef30b} follows from~\eqref{eq:amos_ayef20};
where~\eqref{eq:amos_ayef30c} follows, by the (conditional) data
processing inequality, from
\begin{equation}
  \label{eq:MC2}
  \hat{X} \markov (Y,Z) \markov X
\end{equation} 
(which holds by~\eqref{eq:amos_def_X_hat});
where~\eqref{eq:amos_ayef30d} follows from~\eqref{eq:amos_shikor10}; 
and~\eqref{eq:amos_ayef30e} follows
from~\eqref{eq:amos_ayef10}. Inequalities~\eqref{eq:dir1}
and~\eqref{eq:amos_ayef30e} establish the corollary.

\section{Proof of Proposition~\ref{prop:key_properties}}
\label{app:key_properties}
That $\RitDD$ is bounded by $H(X|Y)$ is just a restatement
of~\eqref{eq:amos_bounded}.  Monotonicity holds because the feasible
set in the minimization defining $\RitDD$ is enlarged (or is
unaltered) when $\Dd$ and/or $\De$ are increased.

As to the convexity, let $Z^{(1)}, \phi^{(1)}, \psi^{(1)}$ and
$Z^{(2)}, \phi^{(2)}, \psi^{(2)}$ be the random variables and
functions that achieve the minima in the definitions of
$\Ritsym\big(D_{\mathrm{d}}^{(1)},D_{\mathrm{e\vphantom{d}}}^{(1)}\big)$
and
$\Ritsym\big(D_{\mathrm{d}}^{(2)},D_{\mathrm{e\vphantom{d}}}^{(2)}\big)$. Let
$Q \sim \text{Bernoulli($\lambda$)}$ be independent of $(X, Y,
Z^{(1)},Z^{(2)})$. Define 
\begin{equation}
 Z \triangleq \bigl(Q,Z^{(Q)} \bigr) 
\end{equation}
and the functions 
\begin{equation}
  \phi(Y,Z) \triangleq \phi^{(Q)}\bigl(Y, Z^{(Q)} \bigr)
\end{equation}
\begin{equation}
  \psi(X,Z) \triangleq \psi^{(Q)}\bigl(X, Z^{(Q)} \bigr).
\end{equation}
Then 
\begin{equation}
 Z \markov X \markov Y;  
\end{equation}
\begin{eqnarray}
\lefteqn{ E[d_{\mathrm{d}}(X,\phi(Y,Z))]}\quad  \\
    &  =   & \lambda E[d_{\mathrm{d}}(X,\phi^{(1)}(Y,Z^{(1)}))] \nonumber \\
    &&
             + (1-\lambda) E[d_{\mathrm{d}}(X,\phi^{(2)}(Y,Z^{(2)}))] \\
    & \leq & \lambda D_{\mathrm{d}}^{(1)} + (1-\lambda) D_{\mathrm{d}}^{(2)};
             \label{equ:cx-dd-leq-lambda}
\end{eqnarray}
and
\begin{eqnarray}
\lefteqn{ E[d_{\mathrm{e}}(\phi(Y,Z),\psi(X,Z))] } \quad \\
    &  =   & \lambda
             E[d_{\mathrm{e}}(\phi^{(1)}(Y,Z^{(1)}),\psi^{(1)}(X,Z^{(1)}))] 
             \nonumber \\
    &      & + (1-\lambda)
             E[d_{\mathrm{e}}(\phi^{(2)}(Y,Z^{(2)}),\psi^{(2)}(X, Z^{(2)}))]\IEEEeqnarraynumspace \\
    & \leq & \lambda D_{\mathrm{e}}^{(1)} + (1-\lambda) D_{\mathrm{e}}^{(2)};
             \label{equ:cx-de-leq-lambda}
\end{eqnarray}
so $Z,\phi, \psi$ are feasible for the distortions
\begin{equation*}
\Bigl(\lambda \Dd^{(1)} + (1-\lambda) \Dd^{(2)} \, , \,
     \lambda \De^{(1)} + (1-\lambda) \De^{(2)}\Bigr).  
\end{equation*}

Consequently,
\begin{IEEEeqnarray}{rCl}
  \IEEEeqnarraymulticol{3}{l}{
    \Ritsym\big(\lambda D_{\mathrm{d}}^{(1)} + (1-\lambda)
  D_{\mathrm{d}}^{(2)}, \lambda D_{\mathrm{e\vphantom{d}}}^{(1)} +
  (1-\lambda) D_{\mathrm{e\vphantom{d}}}^{(2)}\big) }\nonumber\\\quad
& \leq & 
  I(X;Z) - I(Y;Z) \nonumber \\
  & = & H(X) - H(X|Z) - H(Y) + H(Y|Z) \nonumber \\
  &  = & H(X) - H(X|Z^{(Q)},Q) - H(Y) + H(Y|Z^{(Q)},Q) \nonumber \\
  & =  & H(X) - \lambda H(X|Z^{(1)}) - (1-\lambda) H(X|Z^{(2)})
  \nonumber\\
  && -\: H(Y) + \lambda H(Y|Z^{(1)}) + (1-\lambda) H(Y|Z^{(2)})
  \nonumber \\
  & = & \lambda \big(I(X;Z^{(1)}) - I(Y;Z^{(1)})\big) \nonumber\\
  && +\: (1-\lambda) \big(I(X;Z^{(2)}) - I(Y;Z^{(2)})\big).
    \nonumber  \label{equ:cx-i-eq-lambda} \\
    & = & \lambda \, 
    \Ritsym\big(D_{\mathrm{d}}^{(1)},D_{\mathrm{e\vphantom{d}}}^{(1)}\big)
     + (1-\lambda) \,
     \Ritsym\big(D_{\mathrm{d}}^{(2)},D_{\mathrm{e\vphantom{d}}}^{(2)}\big).
\end{IEEEeqnarray}

To conclude the proof it remains to prove that $\RitDD$ is continuous
on $\RealsP^{2}$. (Continuity on $\RealsPP^{2}$ is a consequence of
the convexity, but we also claim continuity in the closed set
$\RealsPP^{2}$.) Since $\RealsP^{2}$ is locally simplicial (as can be
verified by the definition in \cite[Section 10, p.~84]{Rockafellar} or
using \cite[Theorem 20.5, p~184]{Rockafellar}), the convexity of
$\RitDD$ on $\RealsP^{2}$ implies its upper-semicontinuity relative to
$\RealsP^{2}$. It thus remains to prove lower-semicontinuity relative
to $\RealsP^{2}$. That is, we need to show that
\begin{equation*}
 \bigl( D_{\d}^{(\kappa)}, D_{\e}^{(\kappa)} \bigr) \to \bigl(
 D_{\d}, D_{\e} \bigr)
\end{equation*}
implies that there is a subsequence $\{\kappa_{\nu}\}$ such that 
\begin{equation*}
  \Rit{D_{d}}{D_{\e}} \leq 
\lim_{\nu \to \infty} 
\Rit{D_{\d}^{(\kappa_{\nu})}}{D_{\e}^{(\kappa_{\nu})}}.
\end{equation*}
Let $\fdc^{(\kappa)}$, $\fec^{(\kappa)}$, $P_{Z|X}^{(\kappa)}$ achieve
$\Rit{D_{\d}^{(\kappa)}}{D_{\e}^{(\kappa)}}$ with $\mathcal{Z} = \{1,
\ldots, |\mathcal{X}| + 3\}$. Since there are only a finite number of
functions from $\mathcal{Y} \times \mathcal{Z}$ to $\hat{\mathcal{X}}$
and only a finite number of functions from $\mathcal{X} \times
\mathcal{Z}$ to $\hat{\mathcal{X}}$, we can choose a subsequence
$\{\kappa_{\nu}\}$ along which: the mappings $\fdc^{(\kappa_{\nu})}$
do not depend on $\nu$ and can be thus denoted $\fdc$; the mappings
$\fec^{(\kappa_{\nu})}$ do not depend on~$\nu$ and can be thus denoted
$\fec$; and the conditional laws $P_{Z|X}^{(\kappa_{\nu})}$ converge
to some conditional law that we denote $P_{Z|X}^{(0)}$.
By the continuity of mutual information,
$\Rit{D_{\d}^{(\kappa_{\nu})}}{D_{\e}^{(\kappa_{\nu})}}$ converges to
$I(X;Z) - I(Y;Z)$ evaluated with respect to $P_{Z|X}^{(0)} P_{XY}$, and 
$\Rit{D_{d}}{D_{\e}}$ cannot exceed this value because
$P_{Z|X}^{(0)}$, $\fec$, and $\fdc$ are in the feasible set defining it.

\section{Proof of Proposition~\ref{prop:Gach}}\label{sec:Gscheme}
We  present and analyze a scheme that achieves the rate-distortions tuples in Proposition~\ref{prop:Gach}.
Before describing the scheme, we introduce some notation and lemmas on $n$-dimensional spheres.
\subsection{On $n$-dimensional Spheres}
An \emph{$n$-sphere of radius $r>0$ centered at
  $\boldsymbol{\xi}\in\Reals^n$} is the set of all vectors
  $\vect{x} \in \Reals^n$ satisfying 
\begin{equation*}
\|\vect{x}-\boldsymbol{\xi}\| =r.
\end{equation*}
When the center of the sphere $\boldsymbol{\xi}$ is the origin
$\vect{0}$, we call it a \emph{centered} sphere, and when the radius
of the sphere is $1$, we call it a \emph{unit} sphere. 

We denote the angle between two nonzero vectors $\mathbf{u},
\mathbf{v} \in \mathbb{R}^n$ by
$\sphericalangle(\mathbf{u},\mathbf{v})$. Its cosine is
\begin{equation}
  \cos \sphericalangle(\mathbf{u},\mathbf{v}) \triangleq
  \frac{\langle\mathbf{u},\mathbf{v}\rangle}{\|\mathbf{u}\|\|\mathbf{v}\|}.
\end{equation}
Given a nonzero vector $\bfmu$ on an $n$-sphere $\mathcal{S}$, the
\emph{spherical cap of half-angle $\theta$ centered at $\bfmu$} is the
set of all vectors $\vect{x}$ on $\mathcal{S}$ satisfying
\begin{equation*}
\sphericalangle(\bfmu,\mathbf{x}) \geq  \theta.
\end{equation*}
The surface area of such a spherical cap does not depend on the vector
$\boldsymbol{\mu}$ but only on the dimension $n$, the radius of the sphere $r$, and the angle
$\theta$. If the radius $r=1$, we denote this surface area by
$C_n(\theta)$.

We say that a random $n$-vector is uniformly distributed over an $n$-sphere, if it is drawn according to a uniform probability measure over the surface of this sphere.

The proofs of the following four lemmas are based on results in \cite{shannon59} and omitted. 
\begin{lemma}\label{lem:Cnuniform}
  Let $\boldsymbol\Psi$ be uniformly distributed over the centered
  unit $n$-sphere, and let $\boldsymbol\mu$ be a deterministic
  unit-length vector in $\Reals^n$. Then,
\begin{equation}
\Prv{ \inner{\boldsymbol\Psi}{\boldsymbol\mu} \geq \tau}
=\frac{C_n{(\arccos (\tau))}}{C_n(\pi)}, \qquad 0 \leq \tau \leq 1.
\end{equation} 
\end{lemma}

\begin{lemma}\label{lem:Cnlim}
For $0\leq  \tau <1 $:
\begin{equation}
\lim_{n\rightarrow \infty} \frac{1}{n} \log \left(\frac{C_n(\arccos(
    \tau))}{C_n(\pi)} \right) = \frac{1}{2}\log(1-\tau^2).
\end{equation}
\end{lemma}

\begin{lemma}\label{lem:ft}
Let $f\colon \Reals \to(0,1]$ be such that the limit
\begin{equation}
-\eta_1 \triangleq \lim_{n\rightarrow \infty} \frac{1}{n} \log f(n)
\end{equation}
exists and $\eta_1>0$. Then, 
\begin{equation}
\lim_{n\rightarrow \infty} \bigl( 1-f(n) \bigr)^{2^{n\eta_2}} = \begin{cases} 1 &
  \textnormal{if } \eta_1 >\eta_2\\ 0 & \textnormal{if } \eta_1< \eta_2.
\end{cases}
\end{equation}
\end{lemma}

\begin{lemma}\label{lem:lim0}
For $\theta \in (0,\pi/2)$
\begin{equation}\label{eq:lim110}
\lim_{n \rightarrow \infty} \frac{C_n(\theta)}{C_n(\pi)} = 0, 
\end{equation}
whereas for $\theta \in (\pi/2, \pi)$
\begin{equation}\label{eq:lim112}
\lim_{n \rightarrow \infty} \frac{C_n(\theta)}{C_n(\pi)} = 1.
\end{equation}
\end{lemma}

\subsection{Scheme}

Our scheme has parameters
\begin{equation}
a, \;\delta,\; \sigma_W^2 > 0  \quad \textnormal{and} \quad b \geq 0
\end{equation}
that must satisfy
Conditions~\eqref{equ:gaussian-ac-scheme-condition-dec} and
\eqref{equ:gaussian-ac-scheme-condition-enc}, which we repeat for
convenience here:
  \begin{eqnarray}\label{equ:gaussian-ac-scheme-condition-dec2}
  (1-a-b)^2 \sigma_X^2
              + a^2 \sigma_W^2
              + b^2 \sigma_U^2 \leq D_{\mathrm{d}}\;\\
              \label{equ:gaussian-ac-scheme-condition-enc2}
    b^2 \sigma_U^2 \leq D_{\mathrm{e\vphantom{d}}}.
  \end{eqnarray}
  To describe and analyze the scheme we use vector notation. Let
  $\vect{X}$ denote the $n$-dimensional column-vector that results
  when the source symbols are stacked on top of each other
  \begin{equation}
    \vect{X} \triangleq \trans{\begin{pmatrix} X_{1} & X_2 & \ldots & X_n\end{pmatrix}}.
  \end{equation}
  Likewise define the side-information vector $\vect{Y}$ and the
  reconstruction vectors $\vect{\hat{X}}_\d$, and
  $\vect{\hat{X}}_{\e}$.

\subsubsection{Codebook generation}
Let 
\begin{IEEEeqnarray}{rCl}
\sigma_Z^2 & \triangleq &a^2( \sigma_W^2+ \sigma_X^2),\\
  R' &\triangleq & \frac{1}{2}
            \log\left(
              \frac{\sigma_X^2 + \sigma_W^2}{\sigma_W^2}
            \right),\\ 
R &\triangleq & \frac{1}{2}
 \log\left(
    \frac{\sigma_X^2 \sigma_U^2 + \sigma_X^2 \sigma_W^2 +  \sigma_W^2\sigma_U^2 }
         {(\sigma_X^2 + \sigma_U^2)\sigma_W^2}
    \right). 
\end{IEEEeqnarray}
Draw $\lceil 2^{nR'}\rceil$ independent random $n$-vectors
$\{\mathbf{Z}(1),\mathbf{Z}(2),\ldots,\mathbf{Z}(\lceil 2^{nR'}
\rceil)\}$ uniformly over the centered $n$-sphere of radius $r =
\sqrt{n \sigma_Z^2}$. Assign these vectors to $\lfloor 2^{n(R+\delta)}\rfloor$ bins: the first $\lceil 2^{(R'-R-\delta)} \rceil$ are assigned to bin~$1$, the following $\lceil 2^{(R'-R-\delta)} \rceil$ vectors are assigned to bin~$2$, etc. 
More specifically, if 
$\mathcal{B}(m)$ denotes the set of vectors assigned to bin~$m\in\{1,\ldots, \lfloor 2^{n(R+\delta)}\rfloor\}$, then
\begin{equation*}
\mathcal{B}(m) =  \big\{ \vect{Z}_{(m-1)\lceil 2^{(R'-R-\delta)} \rceil+1}, \ldots, \vect{Z}_{m \lceil 2^{(R'-R-\delta)} \rceil}\big\}
\end{equation*} 
for $m=1,\ldots, \lfloor 2^{n(R+\delta)}\rfloor-1$
and 
\begin{equation*}
\mathcal{B}\big( \lfloor 2^{n(R+\delta)}\rfloor \big) \triangleq\big\{ \vect{Z}_{(\lfloor 2^{n(R+\delta)}\rfloor-1)+1}, \ldots, \vect{Z}_{\lceil 2^{nR'}\rceil}\big\}.
\end{equation*}

The codebook $\set{C}\triangleq \{\mathbf{Z}(1),\mathbf{Z}(2),\ldots,\mathbf{Z}(\lceil 2^{nR'} \rceil)\}$.

\subsubsection{Encoder}
Given the source sequence $\mathbf{X}=\mathbf{x}$, the encoder looks for
the codeword $\mathbf{z}^* \in \mathcal{C}$ that is closest to having the
``correct'' angle with $\mathbf{x}$:
\begin{equation}\label{eq:enc2}
  \mathbf{z}^* = 
    \operatorname*{arg\,min}_{ \substack{\mathbf{z} \in \mathcal{C}}}
      \left| \cos \sphericalangle(\mathbf{x},\mathbf{z})
             -
             \sqrt{1 - 2^{-2 R'}}
      \right|.
\end{equation}
The encoder then sends $M=m^*$, where $m^*$ denotes the index of the bin containing $\vect{z}^*$.
It also
produces the reconstruction sequence $\hat{\mathbf{x}}_{\mathrm{e}} = \mathbf{z}^* + b \mathbf{x}$.

\subsubsection{Decoder} Given $M=m^*$ and the side-information vector
$\mathbf{Y}=\mathbf{y}$, the decoder chooses
\begin{equation}\label{eq:dec2}
  \hat{\mathbf{z}} = 
    \operatorname*{arg\,min}_{\mathbf{z} \in\mathcal{B}(m^*)}
      \left| \cos \sphericalangle(\mathbf{y},\mathbf{z})
             - 
             \sqrt{1 - 2^{-2 (R' - R)}}
      \right|,
\end{equation}
and
produces the reconstruction sequence $\hat{\mathbf{x}}_{\mathrm{d}} = \hat{\mathbf{z}} + {b}
\mathbf{y}$.

With probability 1 the $\argmin$s in \eqref{eq:enc2} and \eqref{eq:dec2} are unique.
\subsection{Analysis}

We fix  $\epsilon>0$ sufficiently small such that 
\begin{equation}\label{eq:epssmall}
(1-4\epsilon) \sqrt{1 - 2^{-2 (R' - R)}}  > \sqrt{1 - 2^{-2 (R' - R-\delta/2)}},
\end{equation}
 and define the following four events: 
\begin{enumerate}
\item $ \mathcal{E}_{\mathrm{src}}:\;$ ``The source and
  side information are atypical'',
  i.e., 
  \begin{subequations}
  \begin{IEEEeqnarray}{rCl}
   & \Big|\frac{1}{n}\|\mathbf{X}\|^2 - \sigma_X^2\Big | > \epsilon 
    \sigma_X^2 \quad\text{ or }  \\
    &
    \Big|\frac{1}{n}\|\mathbf{Y}\|^2 - \sigma_Y^2\Big| > \epsilon
    \sigma_Y^2 \quad\text{ or }   \\
    &
    |\cos \sphericalangle(\mathbf{X},\mathbf{Y}) - \rho_{XY}| > \epsilon 
    \rho_{XY}\label{eq:rxy}
  \end{IEEEeqnarray}
  \end{subequations}
  where $\rho_{XY}$ denotes the correlation coefficient between $X$ and $Y$:
  \begin{IEEEeqnarray}{rCl}
  \rho_{XY}=\sqrt{ \frac{ \sigma_X^2}{\sigma_X^2 +\sigma_U^2}}.\IEEEeqnarraynumspace
  \end{IEEEeqnarray}

\item $\mathcal{E}_{\mathrm{enc}}: \;$ ``No codeword has a good angle
  with the source sequence'', i.e.,
  \begin{IEEEeqnarray}{rCl}\Big| \cos \sphericalangle(\mathbf{X},\mathbf{Z^*}) - \sqrt{1-2^{-2R'}}\Big| > \epsilon \sqrt{1-2^{-2R'}}. \IEEEeqnarraynumspace
  \end{IEEEeqnarray}
\item $ \mathcal{E}_{\mathrm{dec}1} :\;$ ``The chosen codeword
  $\vect{Z}^*$ does not have the correct angle with the
  side-information sequence'', i.e.,
  \begin{IEEEeqnarray}{rCl}\hspace{-5mm}\Big| \cos \sphericalangle(\mathbf{Y},\mathbf{Z^*}) - \sqrt{1-2^{-2(R'-R)}}\Big| > 4 \epsilon\sqrt{1-2^{-2(R'-R)}} .\nonumber \\\label{eq:dec1} \IEEEeqnarraynumspace
  \end{IEEEeqnarray}

\item $ \mathcal{E}_{\mathrm{dec}2} :\;$ ``The decoder does not find
  the correct codeword'', i.e.,
  \begin{equation}
\hat{\mathbf{Z}} \neq \mathbf{Z}^*.
  \end{equation}
  \end{enumerate}
%
  
  Also, we define the event 
  \begin{equation*}
  \mathcal{E} \triangleq \mathcal{E}_{\mathrm{src}} \cup \mathcal{E}_{\mathrm{enc}}\cup \mathcal{E}_{\mathrm{dec}1} \cup \mathcal{E}_{\mathrm{dec}2}.
  \end{equation*}

    \begin{lemma}\label{lem:limE}
 \begin{equation}\label{eq:lemE} \lim_{n\to\infty} \Prv{\mathcal{E}}=0.\end{equation}
  \end{lemma}
   \begin{proof}
    We note
  \begin{IEEEeqnarray}{rCl}\label{eq:sumE}
    \mathrm{Pr}[\mathcal{E}]
    & \leq &\mathrm{Pr}[\mathcal{E}_{\mathrm{src}}]
    + \mathrm{Pr}[\mathcal{E}_{\mathrm{enc}} | \mathcal{E}_{\mathrm{src}}^c]
    +
    \mathrm{Pr}[\mathcal{E}_{\mathrm{dec}1}
    | \mathcal{E}_{\mathrm{src}}^c \cap \mathcal{E}_{\mathrm{enc}}^c] \nonumber \\
    & &+  \mathrm{Pr}[\mathcal{E}_{\mathrm{dec}2}
    | \mathcal{E}_{\mathrm{src}}^c \cap \mathcal{E}_{\mathrm{enc}}^c].
  \end{IEEEeqnarray}
  In the following we show that each term on the RHS
  of~\eqref{eq:sumE} tends to zero as the blocklength $n$ tends to
  infinity. The first limit
    \begin{equation}\label{eq:wll1}
    \lim_{n\to\infty} \mathrm{Pr}[\mathcal{E}_{\mathrm{src}}] = 0 \end{equation}
      follows directly from the weak law of large numbers. 
%
The second limit
 \begin{equation}\label{eq:wll2}
\lim_{n\to \infty} \mathrm{Pr}[\mathcal{E}_{\mathrm{enc}}|\mathcal{E}_{\mathrm{src}}^c]=0
\end{equation}
can be shown following the same steps as in the proof of Limit~(134) in  \cite{lapidoth_sending_2010}.
The third limit 
 \begin{equation}\label{eq:wll3}
\lim_{n\to \infty}
    \mathrm{Pr}[\mathcal{E}_{\mathrm{dec}1}
    | \mathcal{E}_{\mathrm{src}}^c \cap \mathcal{E}_{\mathrm{enc}}^c]=0
\end{equation}
is proved as follows.
We have
  \begin{align}\label{eq:coYZ}
    \cos \sphericalangle(\mathbf{Y},\mathbf{Z}^*) 
    & =
    \cos \sphericalangle(\mathbf{X},\mathbf{Y})
    \cos \sphericalangle(\mathbf{X},\mathbf{Z}^*)
    +
    \frac{\langle\mathbf{Y}^{\perp},\mathbf{Z}^{*\perp}\rangle}
    {\|\mathbf{Y}\|\|\mathbf{Z}^*\|}
  \end{align}
where $\mathbf{Y}^{\perp}$ and $\mathbf{Z}^{*\perp}$ denote the components of $\mathbf{Y}$ and $\mathbf{Z}$ that are orthogonal to $\vect{X}$:
  \begin{eqnarray}
    \mathbf{Y}^{\perp} & \triangleq & \vect{Y}- \frac{\inner{\vect{X}}{\vect{Y}}}{\|\vect{X}\|^2} \vect{X}\\& = & 
    \mathbf{Y} -     \cos \sphericalangle(\mathbf{X},\mathbf{Y}) \|\mathbf{Y}\|
    \frac{\mathbf{X}}{\|\mathbf{X}\|}, 
  \end{eqnarray}
  and 
  \begin{eqnarray}
      \mathbf{Z}^{*\perp} & \triangleq & \vect{Z}^*- \frac{\inner{\vect{X}}{\vect{Z}^*}}{\|\vect{X}\|^2} \vect{X}\\& = & 
    \mathbf{Z}^* -    \cos \sphericalangle(\mathbf{X},\mathbf{Z}^*) \|\mathbf{Z}^*\|
    \frac{\mathbf{X}}{\|\mathbf{X}\|}.
      \end{eqnarray}
Let $t_{XZ^*}$ satisfy
\begin{equation}\label{eq:txz}
t_{XZ^*} \in \left[ (1-\epsilon) \sqrt{2^{-2R'}}, (1+ \epsilon) \sqrt{2^{-2R'}}\right]
\end{equation} 
and let $\vect{x}$ and $\vect{y}$ be vectors in $\Reals^n$ satisfying
  \begin{subequations}
  \begin{IEEEeqnarray}{rCl}
   & \Big|\frac{1}{n}\|\mathbf{x}\|^2 - \sigma_X^2\Big| \leq \epsilon \sigma_X^2\\
    &
    \Big|\frac{1}{n}\|\mathbf{y}\|^2 - \sigma_Y^2\Big| \leq \epsilon \sigma_Y^2
    \sigma_Y^2 \\
    &
    |\cos \sphericalangle(\mathbf{x},\mathbf{y}) - \rho_{XY}| \leq \epsilon 
    \rho_{XY}. \label{eq:rhoxy}
  \end{IEEEeqnarray}
  \end{subequations}
Then, conditional on  events
\begin{equation}\label{eq:condevents}
  \mathcal{E}_{\mathrm{src}}^c,\;\;\mathcal{E}_{\mathrm{enc}}^c,\;\;
        \mathbf{X}=\mathbf{x},\;\;\mathbf{Y}=\mathbf{y},\; \;\cos \sphericalangle(\mathbf{X},\mathbf{Z^*})=t_{XZ^*},
\end{equation}
by~\eqref{eq:txz} and \eqref{eq:rhoxy}, we have
\begin{subequations}\label{eq:product}
\begin{IEEEeqnarray}{rCl}
    \cos \sphericalangle(\mathbf{X},\mathbf{Y})
    \cos \sphericalangle(\mathbf{X},\mathbf{Z}^*) &\leq &(1+\eps) \rho_{XY}(1+ \epsilon) \sqrt{2^{-2R'}}\nonumber \\
    &\stackrel{(a)}{\leq} & \sqrt{1- 2^{-(R'-R)}} (1+3\epsilon)\IEEEeqnarraynumspace
\end{IEEEeqnarray}
and 
\begin{IEEEeqnarray}{rCl}
    \cos \sphericalangle(\mathbf{X},\mathbf{Y})
    \cos \sphericalangle(\mathbf{X},\mathbf{Z}^*) &\geq &(1-\eps) \rho_{XY}(1- \epsilon) \sqrt{2^{-2R'}}\nonumber \\
    & \stackrel{(a)}{\geq} & \sqrt{1- 2^{-(R'-R)}} (1-3\epsilon), \nonumber \\\IEEEeqnarraynumspace
\end{IEEEeqnarray}
\end{subequations}
where Inequalities~$(a)$ follow because 
\begin{equation}
\rho_{XY} \cdot\sqrt{1- 2^{-2R'}}=\sqrt{1-2^{-(R'-R)}}
\end{equation} and because $\epsilon\in(0,1)$.
Moreover, conditional on the events in~\eqref{eq:condevents}, 
  the vector $\mathbf{Z}^{*\perp}$ is uniformly distributed over a centered $(n-1)$-dimensional sphere of radius 
  $\sigma_Z^2(1-t_{XZ^*}^2)$, and thus  Limit~\eqref{eq:limperp} on top of the next page follows by Lemmas~\ref{lem:Cnuniform} and \ref{lem:lim0}.
  
We can combine Limit~\eqref{eq:limperp} and Inequalities~\eqref{eq:product} to obtain the limit~\eqref{eq:limitE2}  on top of the next page. 
\begin{figure*}
    \begin{equation}\label{eq:limperp}
    \lim_{n \rightarrow \infty}
    \mathrm{Pr}\left[
      \left|
    \langle{\mathbf{y}}^{\perp},{\mathbf{Z}}^{*\perp}\rangle
       \right|
       \leq  \epsilon \sqrt{1- 2^{-2(R'-R)}}   \|\vect{y}\|\sqrt{\sigma_Z^2}
        \Big|
        \mathbf{X}=\mathbf{x}, \mathbf{Y}=\mathbf{y}, \cos \sphericalangle(\mathbf{X},\mathbf{Z^*})=t_{XZ^*}
    \right]
    = 1
  \end{equation}
  \begin{equation}\label{eq:limitE2}
    \lim_{n \rightarrow \infty}
    \mathrm{Pr}\left[
      \Big|  \cos \sphericalangle(\mathbf{Y},\mathbf{Z}^*) - \sqrt{1- 2^{-2(R'-R)}}  \Big|
       \leq4 \epsilon \sqrt{1- 2^{-2(R'-R)}}   \Big|
       \mathcal{E}_{\mathrm{src}}^c,\mathcal{E}_{\mathrm{enc}}^c,
        \mathbf{X}=\mathbf{x},\mathbf{Y}=\mathbf{y}, \cos \sphericalangle(\mathbf{X},\mathbf{Z^*})=t_{XZ^*}
    \right]
    = 1
  \end{equation}
    \hrule 
   \end{figure*}
If in~\eqref{eq:limitE2} we take the expectation with respect to $\vect{X}, \vect{Y}$, and $\cos \sphericalangle(\mathbf{X},\mathbf{Z^*})$ (but keep the conditioning on events $\mathcal{E}_{\textnormal{src}}^{c}$ and $\mathcal{E}_{\textnormal{enc}}^{c}$), we obtain the desired third limit~\eqref{eq:wll3}. 

%
We finally prove the fourth limit 
 \begin{equation}\label{eq:wll4}
\lim_{n\to \infty}
    \mathrm{Pr}[\mathcal{E}_{\mathrm{dec}2}
    | \mathcal{E}_{\mathrm{src}}^c \cap \mathcal{E}_{\mathrm{enc}}^c]=0.
\end{equation}
To this end, we define event $\mathcal{E}_2$ as
\begin{equation}
   \cos \sphericalangle(\mathbf{Y},\mathbf{Z}')
          < \sqrt{1 - 2^{-2 (R' - R-\delta/2)}}, \quad \forall \mathbf{Z}' \in \left(\mathcal{B}(M) \backslash \mathbf{Z}^*\right).
\end{equation} 
Recalling the decoding rule in~\eqref{eq:dec2} and the definition of event $\mathcal{E}_{\textnormal{dec}1}$ in~\eqref{eq:dec1}, we see that  when $\mathcal{E}_{\mathrm{dec}1}^{c}$ and $\mathcal{E}_2$
occur simultaneously, then by condition~\eqref{eq:epssmall} the decoder finds the correct codeword $\hat{\vect{Z}}=\vect{Z}^*$. Therefore, 
 \begin{IEEEeqnarray}{rCl}\label{eq:sumU}
\Prv{\mathcal{E}_{\mathrm{dec}2}|\mathcal{E}_{\mathrm{src}}^c,\mathcal{E}_{\mathrm{enc}}^c} & \leq &1 -  \Prv{\mathcal{E}_{\mathrm{dec}1}^c \cap \mathcal{E}_2|\mathcal{E}_{\mathrm{src}}^c,\mathcal{E}_{\mathrm{enc}}^c}, \IEEEeqnarraynumspace
\end{IEEEeqnarray}  
and thus~\eqref{eq:wll3} and the limit
  \begin{equation}\label{equ:gaussian-ac-er-other-closer-pr-others-close}      \lim_{n \to \infty}    \Prv{\mathcal{E}_{2}^c |\mathcal{E}_{\mathrm{src}}^c,\mathcal{E}_{\mathrm{enc}}^c}=0
    \end{equation}
   establish~\eqref{eq:wll4}.
   
  We now prove~\eqref{equ:gaussian-ac-er-other-closer-pr-others-close}. 
For each $m\in\big\{1,\ldots, \lfloor 2^{n(R+\delta)} \rfloor\big\}$, we index the vectors in the $m$-th bin from $1$ to $|\set{B}(m)|$ and we shall refer to the $k$-th vector in this  $m$-th bin by $\vect{Z}_{m,k}$. Let $K^*$ be the index of $\vect{Z}^*$, i.e., $\vect{Z}_{M,K^*}=\vect{Z}^*$. 
  By the symmetry of the code construction and the encoding rule, the probability $\Prv{\set{E}^{c}| \mathcal{E}_{\mathrm{src}}^c,\mathcal{E}_{\mathrm{enc}}^c, M=m, K^*=k}$ does not depend on the values $m$ and $k$. We therefore, assume in the following that $M=1$ and $K^*=1$.  If we additionally condition on $\vect{X}=\vect{x}$ and on $\cos \sphericalangle(\mathbf{X},\mathbf{Z^*})=t_{XZ^*}>0$, the vectors $\vect{Z}_{1,2}, \ldots, \vect{Z}_{1,|\set{B}(1)|}$ (i.e., the vectors in bin~$1$  that are not $\vect{Z}^*$) are independent and uniformly distributed over the centered $n$-sphere of radius $\sqrt{n\sigma_Z^2}$  without the spherical cap of half-angle $\arccos(t_{XZ^*})$ centered at $\vect{x}$. Thus, $\frac{2}{C_n(\pi)}$ is an upper bound on the conditional density of the normalized vectors  $\frac{1}{\sqrt{n\sigma_Z^{2}}}\vect{Z}_{1,2}, \ldots, \frac{1}{\sqrt{n\sigma_Z^2}}\vect{Z}_{1,|\set{B}(1)|}$ on the centered unit $n$-sphere. Applying Lemma~\ref{lem:Cnuniform},  we 
 therefore obtain Inequality~\eqref{equ:gaussian-ac-er-other-closer-cond-pr-others-close}  shown on top of the next page.  \begin{figure*}
  \begin{IEEEeqnarray}{rCl}\label{eq:In1}
  \lefteqn{\mathrm{Pr}\left[\bigcup_{ k=2}^{|\mathcal{B}(1)|}
    \left(
          \cos \sphericalangle(\mathbf{Y},\mathbf{Z}_{1,k})
          \geq \sqrt{1 - 2^{-2 (R' - R-\delta/2)}}
    \right)
    \left|
    \begin{aligned}
      & \mathbf{X} = \mathbf{x},M=1, K^*=1,
      \mathcal{E}_{\mathrm{src}}^c,\mathcal{E}_{\mathrm{enc}}^c
    \end{aligned}
    \right.
    \right] }\qquad \nonumber  \\
    &  =&
1 -\prod_{ k=2}^{|\mathcal{B}(1)|} \left(1 - \mathrm{Pr}\left[
               \cos \sphericalangle(\mathbf{Y},\mathbf{Z}_{1,k})
               \geq \sqrt{1 - 2^{-2 (R' - R-\delta/2)}}
              \left|
    \begin{aligned}
      & \mathbf{X} = \mathbf{x}, M=1,
      K^*=1,\mathcal{E}_{\mathrm{src}}^c,\mathcal{E}_{\mathrm{enc}}^c
    \end{aligned}
    \right.
             \right]\right)\IEEEeqnarraynumspace
\\&&  <    1 -        
                      \left(1 -
               \frac{2 \, C_n(\arccos(\sqrt{1 - 2^{-2 (R' - R-\delta/2)}}))}{C_n(\pi)}\right)^{|\set{B}(1)|-1}
       \label{equ:gaussian-ac-er-other-closer-cond-pr-others-close}
       \\
       & & \leq 1 -
    \left(1-\frac{2 \, C_n(\arccos(\sqrt{1 - 2^{-2 (R' - R-\delta/2)}}))}{C_n(\pi)} \right)^{2^{n (R' - R-\delta)}}
    \label{equ:gaussian-ac-er-other-closer-pr-others-close-bound}
  \end{IEEEeqnarray}
  \hrule
  \end{figure*}
  We note
  that for any $\gamma\in[0,1]$ 
  \begin{equation}0\leq \left(1-\frac{2 \, C_n(\arccos(\gamma))}{C_n(\pi)} \right) \leq 1
  \end{equation}
  and hence the  mapping
  $t \mapsto \left(1-\frac{2 \, C_n(\arccos(\gamma))}{C_n(\pi)} \right)^t$
  is decreasing in $t>0$. Therefore, since 
  \begin{equation}
  |\mathcal{B}(1)|-1 <  2^{n (R' - R-\delta)}
  \end{equation} we further obtain~\eqref{equ:gaussian-ac-er-other-closer-pr-others-close-bound}. 
  If now we take the expectation with respect to $\mathbf{X}$, $M$, and $K^*$ (but keep the conditioning on $ \mathcal{E}_{\mathrm{src}}^c$ and $\mathcal{E}_{\mathrm{enc}}^c$), \eqref{equ:gaussian-ac-er-other-closer-pr-others-close-bound} results in
  \begin{align}
    & \mathrm{Pr}\left[\left.\mathcal{E}_2
    \right|
    \mathcal{E}_{\mathrm{src}}^c,\mathcal{E}_{\mathrm{enc}}^c
    \right] \nonumber\\
    &  <
    1 -
    \left(1-\frac{2 \, C_n\left(\arccos\sqrt{1 - 2^{-2 (R' - R-\frac{\delta}{2})}}\right)}{C_n(\pi)} \right)^{2^{n (R' - R-\delta)}}\!\!.
       \label{equ:gaussian-ac-er-other-closer-pr-others-close-bound2}
  \end{align}
The desired limit~(\ref{equ:gaussian-ac-er-other-closer-pr-others-close}) follows by \eqref{equ:gaussian-ac-er-other-closer-pr-others-close-bound2} and by Lemma~\ref{lem:ft}. In fact, applying Lemma~\ref{lem:ft} to 
\begin{equation}\label{eq:eta2}
\eta_2=R' - R-\delta
\end{equation} and to the function 
\begin{equation}
f\colon n \to \frac{2 \, C_n(\arccos(
      \sqrt{1 - 2^{-2 (R' - R-\delta/2)}}))}{C_n(\pi)},
\end{equation}
 we obtain that the right-hand side of~\eqref{equ:gaussian-ac-er-other-closer-pr-others-close-bound2} tends to 1 as $n$ tends to infinity because
  \begin{IEEEeqnarray}{rCl}\label{equ:gaussian-ac-er-other-closer-lim-log-cn}
  \eta_1&\triangleq &- \lim_{n\rightarrow \infty} \frac{1}{n} \log \left(\frac{2 \, C_n(\arccos(
      \sqrt{1 - 2^{-2 (R' - R-\delta/2)}}))}{C_n(\pi)} \right) \nonumber \\ & =& R' - R-\delta/2\\
      & >& \eta_2. \hspace{4cm}
  \end{IEEEeqnarray}
  Here, the equality holds by Lemma~\ref{lem:Cnlim} and because the factor 2 in the logarithm does not change the limit, and the inequality holds by~\eqref{eq:eta2} and because $\delta>0$.
%

This concludes the proof of limit~(\ref{equ:gaussian-ac-er-other-closer-pr-others-close}) and thus of the fourth limit~\eqref{eq:wll4}. 
Combining finally~\eqref{eq:sumE} with \eqref{eq:wll1}--\eqref{eq:wll3} and \eqref{eq:wll4}  establishes the proof of the lemma.
  \end{proof}

%
We can now bound the expected distortions of our scheme. We have
    \begin{IEEEeqnarray}{rCl}\label{equ:gaussian-ac-exp-dd}
    \E{\distn{\mathrm{d}}(\mathbf{X},\hat{\mathbf{X}}_{\mathrm{d}})} & 
    =& 
   \Prv{\mathcal{E}^c} \E{\distn{\mathrm{d}}(\mathbf{X},\hat{\mathbf{X}}_{\mathrm{d}})\big|\mathcal{E}^c} \nonumber\\
    && 
    + 
    \Prv{\mathcal{E}} \E{\distn{\mathrm{d}}(\mathbf{X},\hat{\mathbf{X}}_{\mathrm{d}})\big|\mathcal{E}},\IEEEeqnarraynumspace
  \end{IEEEeqnarray}
  and
  \begin{IEEEeqnarray}{rCl}\label{equ:gaussian-ac-exp-de}
    \E{\distn{\mathrm{e}}(\hat{\mathbf{X}}_{\mathrm{d}},\hat{\mathbf{X}}_{\mathrm{e}})}
   & =&
    \mathrm{Pr}[\mathcal{E}^c]
    \E{\distn{\mathrm{e}}(\hat{\mathbf{X}}_{\mathrm{d}},\hat{\mathbf{X}}_{\mathrm{e}})\big|\mathcal{E}^c}\nonumber \\
    && + 
    \mathrm{Pr}[\mathcal{E}]
    \E{\distn{\mathrm{e}}(\hat{\mathbf{X}}_{\mathrm{d}},\hat{\mathbf{X}}_{\mathrm{e}})\big|\mathcal{E}}.\IEEEeqnarraynumspace
  \end{IEEEeqnarray}
The decoder-side distortion satisfies
\begin{IEEEeqnarray}{rCl}
d_{\mathrm{d}}^{(n)}(\mathbf{x},\hat{\mathbf{x}}_{\mathrm{d}})
     & = &\frac{1}{n} \|\mathbf{x} - \mathbf{z}^* - {b}\vect{y} \|^2\\
     & \leq & \frac{3}{n}  \|\mathbf{x}\|^2 +\frac{3}{n} \|\mathbf{z}^* \|^2 + \frac{3}{n} {b^2} \|\vect{y}\|^2,\label{eq:ddrough}
\end{IEEEeqnarray}
where the inequality holds by the Cauchy-Schwarz Inequality and because an arithmetic mean of two nonnegative numbers cannot be smaller than it's geometric mean. 
Therefore, 
    \begin{IEEEeqnarray}{rCl}\lefteqn{
    \Prv{\mathcal{E}} \E{\distn{\mathrm{d}}(\mathbf{X},\hat{\mathbf{X}}_{\mathrm{d}})\big|\mathcal{E}}} \; \nonumber \\& \leq &   \frac{3}{n}  \Prv{\mathcal{E}} \E{   \|\mathbf{X}\|^2 + \|\mathbf{Z}^* \|^2 +  {b^2} \|\vect{Y}\|^2 \big|\mathcal{E}} \\
    & = & \frac{3}{n}\E{  \|\mathbf{X}\|^2 +\|\mathbf{Z}^* \|^2 +{b^2} \|\vect{Y}\|^2 }\nonumber \\ & & -    \frac{3}{n}  \Prv{\mathcal{E}^c} \E{  \|\mathbf{X}\|^2 +\|\mathbf{Z}^* \|^2 +{b^2} \|\vect{Y}\|^2 \big|\mathcal{E}^c} \IEEEeqnarraynumspace \\
&\leq	 &     3\Big(\sigma_X^2+\sigma_Z^2 +{b^2}( \sigma_X^2+\sigma_U^2)\Big)\nonumber \\ & & -   3\big( \sigma_X^2(1-\epsilon) +\sigma_Z^2 +{b^2}( \sigma_X^2+\sigma_U^2)(1-\epsilon)\Big) \Prv{\mathcal{E}^c} \\
& \leq&  3\Big(\sigma_X^2+\sigma_Z^2 +{b^2}( \sigma_X^2+\sigma_U^2)\Big)\big( 1- (1-\epsilon) \Prv{\mathcal{E}^c}\big).\label{eq:ddE}
    \IEEEeqnarraynumspace
  \end{IEEEeqnarray}
In the event $\mathcal{E}^c$,  we can derive a bound on the decoder-side distortion $d_\d^{(n)}(\mathbf{x},\hat{\mathbf{x}}_{\mathrm{d}})$ that is tighter than~\eqref{eq:ddrough}:
  \begin{align}
    & d_{\mathrm{d}}^{(n)}(\mathbf{x},\hat{\mathbf{x}}_{\mathrm{d}}) \nonumber\\
     &\quad = \frac{1}{n} \|\mathbf{x} - \mathbf{z}^* - {b}
     \mathbf{y}\|^2 \\
     &\quad = \frac{1}{n} \|\mathbf{x}\|^2 +\frac{1}{n}  \|\mathbf{z}^*\|^2  + \frac{b^2}{n} \|\mathbf{y}\|^2\nonumber \\
     & \quad \quad -\frac{2}{n} \langle \mathbf{x}, \mathbf{z}^* \rangle
     - \frac{2b}{n}  \langle \mathbf{x}, \mathbf{y} \rangle
+ \frac{2b}{n} \langle \mathbf{z}^*, \mathbf{y} \rangle
     \\
     &\quad \leq (1+\epsilon) \sigma_X^2 +\sigma_Z^2 
     + (1+\epsilon) b^2 (\sigma_X^2 + \sigma_U^2)
 \nonumber \\
     &\quad \quad     - 2(1-\epsilon)^2 a \sigma_X^2 
     - 2 (1-\epsilon)^3 b \sigma_X^2
     \nonumber \\ &\quad \quad   + 2 (1+\epsilon)(1+4 \epsilon) ab \sigma_X^2 \\
     &\quad \leq
     (1+a^2+b^2-2a-2b+2ab)\sigma_X^2 
     + a^2 \sigma_W^2
     +b^2  \sigma_U^2 \nonumber \\
     &\quad \quad
     +\epsilon (\sigma_X^2  
     + b^2 (\sigma_X^2 + \sigma_U^2)
     + 4 a \sigma_X^2
     + 6  b \sigma_X^2+10  ab \sigma_X^2
     ) 
    \nonumber \\
     &\quad \quad  +8  \epsilon^2 ab \sigma_X^2 +2 \epsilon^3 b\sigma_X^2 
     \\
     &\quad \leq
              D_{\mathrm{d}}\nonumber \\ & \qquad 
     +\epsilon (\sigma_X^2  
     + b^2 (\sigma_X^2 + \sigma_U^2)
     + 4 a \sigma_X^2
     + 8  b \sigma_X^2
     + 18  ab \sigma_X^2)
  \end{align}
  where the first inequality follows from the definition of the event $\mathcal{E}^c$, the second by throwing away some negative $\epsilon$-terms,  and the third from
   Condition~\eqref{equ:gaussian-ac-scheme-condition-dec2} and because $\epsilon<1$. Since $  \mathrm{Pr}[\mathcal{E}^c] \leq 1$, we thus have: 
   \begin{IEEEeqnarray}{rCl}\label{eq:ddEc}
\lefteqn{\mathrm{Pr}[\mathcal{E}^c]
    \E{\distn{\mathrm{e}}(\hat{\mathbf{X}}_{\mathrm{d}},\hat{\mathbf{X}}_{\mathrm{e}})\big|\mathcal{E}^c}}\nonumber \\
& \leq& D_d
     +\epsilon(\sigma_X^2  
     + b^2 (\sigma_X^2 + \sigma_U^2)
     + 4 a \sigma_X^2
     + 8  b \sigma_X^2
     + 18  ab \sigma_X^2).\nonumber \\
   \end{IEEEeqnarray}
Combining \eqref{equ:gaussian-ac-exp-dd}, \eqref{eq:ddE}, and \eqref{eq:ddEc}, we obtain
\begin{IEEEeqnarray}{rCl}
\lefteqn{\E{d_\textnormal{d}^{(n)}(\vect{X}, \hat{\vect{X}}_{\textnormal{d}})} }\quad \\ &\leq &D_{\d}+ 3\Big(\sigma_X^2+ \sigma_Z^2+ {b^2}\sigma_Y^2\Big)\big(1- (1+\epsilon)\Prv{\set{E}^{c}} \big)   \nonumber \\ & & +\epsilon (\sigma_X^2  
     + b^2 (\sigma_X^2 + \sigma_U^2)
     + 4 a \sigma_X^2
     + 8  b \sigma_X^2
     + 18  ab \sigma_X^2).
     \nonumber \\
 \label{eq:concDd}
\end{IEEEeqnarray}

Similarly, we have for the encoder-side distortion:
\begin{IEEEeqnarray}{rCl}
d_{\mathrm{e}}^{(n)}(\mathbf{x},\hat{\mathbf{x}}_{\mathrm{d}})
     & = &\frac{1}{n}\left \|{b}\vect{y}-b\vect{x} \right\|^2\\
     & \leq & \frac{2}{n} {b^2} \|\vect{y}\|^2  + \frac{2}{n}b^2 \|\mathbf{x}\|^2, \label{eq:derough}
\end{IEEEeqnarray}
and thus, 
    \begin{IEEEeqnarray}{rCl}\lefteqn{
    \Prv{\mathcal{E}} \E{\distn{\mathrm{e}}(\mathbf{X}_\d,\hat{\mathbf{X}}_{\mathrm{e}})\big|\mathcal{E}}} \qquad \nonumber \\
    & \leq & \frac{2}{n}\E{  {b^2}\|\vect{Y}\|^2  +b^2 \|\mathbf{X}\|^2 }\nonumber \\ & & -    \frac{2}{n}  \Prv{\mathcal{E}^c} \E{{b^2} \|\vect{Y}\|^2  +b^2 \|\mathbf{X}\|^2 \Big|\mathcal{E}^c} \IEEEeqnarraynumspace \\
& \leq&   2\Big({b^2} (\sigma_X^2+\sigma_U^2)+ b^2 \sigma_X^2\Big)\big( 1- (1-\epsilon) \Prv{\mathcal{E}^c}\big).\label{eq:deE}
    \IEEEeqnarraynumspace
  \end{IEEEeqnarray}
Moreover, in the event $\mathcal{E}^c$ we can derive a bound on the encoder-side distortion $d_\e^{(n)}(\hat{\mathbf{x}}_{\mathrm{d}},\hat{\mathbf{x}}_{\mathrm{e}})$ that is tighter than~\eqref{eq:derough}:
  \begin{align}
    d_{\mathrm{e}}^{(n)}(\hat{\mathbf{x}}_{\mathrm{d}},\hat{\mathbf{x}}_{\mathrm{e}})
     &  = \frac{1}{n} \left\|{b} \mathbf{y}
      - b \mathbf{x}\right\|^2 \\
     &  = \frac{1}{n} b^2 \Big( \|\mathbf{x}\|^2 +
  \|\mathbf{y}\|^2
     - 2\langle \mathbf{x}, \mathbf{y} \rangle
     \Big)\\
     & \leq 
     (1+\epsilon)b^2\sigma_X^2 + (1+\epsilon) b^2(\sigma_X^2 +\sigma_U^2)  
\nonumber \\ & \quad     - 2b^2 (1-\epsilon)^3 \sigma_X^2
    \\
     & \leq b^2 \sigma_U^2
     + \epsilon b^2 (8\sigma_X^2 + \sigma_U^2)+ \epsilon^3 b^2 \sigma_X^2
     \\
     & \leq 
              D_{\mathrm{e}}
              + \epsilon b^2 (9\sigma_X^2 + \sigma_U^2),
  \end{align}
  where the last inequality follows by Assumption~\eqref{equ:gaussian-ac-scheme-condition-enc2} and because $\epsilon<1$. Since $\mathrm{Pr}[\mathcal{E}^c]\leq 1$, we thus have  
  \begin{IEEEeqnarray}{rCl}\label{eq:deEc}
  {\mathrm{Pr}[\mathcal{E}^c]
    \E{\distn{\mathrm{e}}(\hat{\mathbf{X}}_{\mathrm{d}},\hat{\mathbf{X}}_{\mathrm{e}})\big|\mathcal{E}^c}} 
         & \leq &
               D_{\mathrm{e}}
              + \epsilon b^2 (9\sigma_X^2 + \sigma_U^2). \IEEEeqnarraynumspace
    \end{IEEEeqnarray}
Combining finally~\eqref{equ:gaussian-ac-exp-de}, \eqref{eq:deE}, and \eqref{eq:deEc}, we obtain
\begin{IEEEeqnarray}{rCl}\lefteqn{
\E{\distn{\mathrm{e}}(\mathbf{X}_\d,\hat{\mathbf{X}}_{\mathrm{e}})} }\nonumber \quad \\ & \leq &   D_{\mathrm{e}} +2\Big({b^2} \sigma_Y^2+ b^2 \sigma_X^2\Big)\big( 1- (1-\epsilon) \Prv{\mathcal{E}^c}\big)  \nonumber \\ &&+ \epsilon b^2 (9\sigma_X^2 + \sigma_U^2). \label{eq:concDe}
\end{IEEEeqnarray}

Recall that the rate of our scheme is smaller than $R+\delta$ and that $\epsilon, \delta>0$ can be chosen arbitrarily close to 0. Therefore, from \eqref{eq:concDd}, \eqref{eq:concDe}, and Lemma~\ref{lem:limE} we  conclude that when $a, \sigma_W^2>0$ and $b\geq 0$ satisfy~\eqref{equ:gaussian-ac-scheme-condition-dec2} and \eqref{equ:gaussian-ac-scheme-condition-enc2}, then our scheme can achieve  
the triple 
\begin{equation}
\left(R=
 \frac{1}{2}
 \log\left(
    \frac{\sigma_X^2 \sigma_U^2 + \sigma_X^2 \sigma_W^2 + \sigma_U^2 \sigma_W^2}
         {(\sigma_X^2 + \sigma_U^2)\sigma_W^2}
    \right),
 D_{\mathrm{d}},D_{\mathrm{e\vphantom{d}}}
\right).
\end{equation} This establishes Proposition~\ref{prop:Gach}.

\section{The Cardinality Bound on $\set{U}$}
\label{app:Amos}
To prove the cardinality bound~\eqref{eq:amos_cardU} on $\set{U}$, we
shall need the following variation on Carath\'eodory's theorem.
\begin{lemma}\label{lem:card}
  Any point on the boundary of the convex hull of a compact set in
  $\Reals^d$ can be expressed as a convex combination of $d$ or fewer
  points in the set.
\end{lemma}
\begin{proof}
  Let $\set{S}$ be a compact subset of $\Reals^{d}$, and let
  $\vect{x}$ be a boundary point of its convex hull
  $\textnormal{conv}(\set{S})$. Since $\vect{x}$ is in the convex hull
  of $\set{S}$, it follows from Carath\'eodory's theorem that there
  exist $d+1$ or fewer points
  \begin{equation}
    \vect{x}_{1}, \ldots, \vect{x}_{\nu} \in \set{S}, \qquad \nu \leq d+1
  \end{equation}
  and positive coefficients summing to $1$
  \begin{equation}
    \label{eq:tahina5}
    \lambda_{1}, \ldots, \lambda_{\nu} > 0, \qquad \sum_{i=1}^{\nu}
    \lambda_{i} = 1
  \end{equation}
  such that
  \begin{equation}
    \label{eq:tahina10}
    \vect{x}= \sum_{i=1}^{\nu}\lambda_i \, \vect{x}_i.
  \end{equation}
  We shall show that, in fact, of these $\nu$ points, we can find $d$
  or fewer points whose convex combination is $\vect{x}$.

  Since $\vect{x}$ is on the boundary of $\textnormal{conv}(\set{S})$,
  there exists a hyperplane $\set{H}$ that supports
  $\textnormal{conv}(\set{S})$ at $\vect{x}$.  Thus,
  \begin{subequations}
  \begin{equation}
    \set{H} = \bigl\{ \boldsymbol{\xi} \in \Reals^{d} \colon
    \trans{\vect{c}} \boldsymbol{\xi} = \trans{\vect{c}} \vect{x} \bigr\}
  \end{equation}
  for some vector $\vect{c} \in \Reals^{d}$ and
  \begin{equation}
    \trans{\vect{c}} \vect{{x}}= 
    \max_{\vect{\tilde{x}}\in \textnormal{conv}(\set{S})} \trans{\vect{c}}
    \vect{\tilde{x}}
  \end{equation}
  \end{subequations}
  so
  \begin{equation}
    \label{eq:tahina20}
    \trans{\vect{c}} \vect{x} \geq \trans{\vect{c}} \vect{x}_{i},
    \quad i = 1, \ldots, \nu.
  \end{equation}
  We shall next show that the points $\vect{x}_{1}, \ldots,
  \vect{x}_{\nu}$ are in $\set{H}$. To that end we note that 
  by~\eqref{eq:tahina10}
  \begin{align*}
    0 & =  \trans{\vect{c}} \biggl( \vect{x} - \sum_{i=1}^{\nu}\lambda_i
      \, \vect{x}_i \biggr) \\
    & = \sum_{i=1}^{\nu} \lambda_{i} \trans{\vect{c}}
    \vect{x} - \sum_{i=1}^{\nu} \lambda_{i} \trans{\vect{c}}
    \vect{x}_i \\
    & = \sum_{i=1}^{\nu} \lambda_{i} \Bigl( \trans{\vect{c}}
    \vect{x} - \trans{\vect{c}}
    \vect{x}_i \Bigr)
  \end{align*}
  where the second equality holds because the $\lambda$'s sum to $1$
  \eqref{eq:tahina5}. Since the $\lambda$'s are all positive, it
  follows from~\eqref{eq:tahina20} that all the terms on the RHS are
  nonnegative. Since they sum to zero, they must all be zero. And
  since the $\lambda$'s are positive, we conclude that
  \begin{equation}
    \label{eq:tahina30}
    \trans{\vect{c}} \vect{x}_i = \trans{\vect{c}} \vect{x}, \quad i
    \in\{1, \ldots, \nu\}
  \end{equation}
  and the vectors $\vect{x}_{i}$ are all in $\set{H}$. The vector
  $\vect{x}$ can thus be written as a convex combination of the $\nu$
  vectors in $\vect{x}_{1}, \ldots, \vect{x}_{\nu}$ in $\set{H}$. Since
  $\set{H}$ is $(d-1)$-dimensional, it follows from Carath\'eodory's
  theorem that $\vect{x}$ is in fact a convex combination of $d$ or
  fewer of the vectors $\vect{x}_{1}, \ldots, \vect{x}_{\nu}$.
\end{proof} 

The cardinality bound on $\set{U}$ can now be proved as follows. 
\begin{proof}[Proof of the Cardinality Bound on $\set{U}$ in
  Proposition~\ref{prop:AmosCardBound}]
Let the discrete random variables $U$ and $Z$ over the alphabets
$\set{U}$ and $\set{Z}$, the function $\phi\colon \set{Y} \times
\set{Z} \to \hat{\set{X}}_{\d}$, and the function $\psi\colon \set{X}
\times \set{Z} \times \set{U} \to \hat{\set{X}}_\e$
satisfy~\eqref{eq:amos4001} and \eqref{eq:distErk}. We shall exhibit a
random variable $\tilde{U}$ over the alphabet
   \begin{equation}
  \tilde{\set{U}}\triangleq \{1,\ldots, K\}
  \end{equation}  
  and a  function $\tilde{\psi}\colon \set{X}\times \set{Z} \times \tilde{\set{U}} \to \hat{\set{X}}_\e$ satisfying
  \begin{equation}\label{eq:MarkovU} 
  \tilde{U} \markov ({X}, {Z}) \markov Y
  \end{equation} 
  and the $K$ distortion constraints
   \begin{equation}\label{eq:ineqDU}
     \E{ d_k\bigl(X, \fdc(Y,Z), \tilde{\psi}(X,Z,\tilde{U})\bigr)} \leq D_{k}, 
     \quad k\in\{1,\ldots, K\}.
  \end{equation}
  Since the Markov conditions~\eqref{eq:amos4001} and \eqref{eq:MarkovU} imply~
  \begin{equation}
  (\tilde{U}, Z) \markov X \markov Y,
  \end{equation}
  this will allow us to replace $U$ and $\psi$ with $\tilde{U}$ and
  $\tilde{\psi}$ and thus conclude the proof.

  To describe  $\tilde{U}$ and $\tilde{\psi}$, we need some
  definitions.  For each pair $(x,z)\in \set{X} \times \set{Z}$
  and each $k\in\{1,\ldots, K\}$, define
  \begin{align}
    D_{k}^{(x,z)} & = \BigPrvcond{d_k\bigl(X, \fdc(Y,Z),\psi(X,Z,U)
      \bigr)}{(X,Z) = (x,z)} \nonumber \\
    & = \E{ d_k \bigl( x, \fdc(Y,z),\psi(x,z,U) \bigr)}, 
  \label{eq:dkxz}
\end{align} 
where the expectation is, by~\eqref{eq:amos4001}, with respect to
$P_{U|XZ}(\cdot|x,z) \, P_{Y|X}(\cdot|x)$. Define also the vector-valued
function
  \begin{IEEEeqnarray}{rCl}\label{eq:funh}
h^{(x,z)} \colon \set{U}  &\to& \Reals^K_+\nonumber \vspace{2mm}\\ u &\mapsto&  \begin{pmatrix}  \E{ {d}_1\bigl(x,  \fdc(Y,z), \psi(x,z,u)\bigr)}\\ \vdots\\ \E{ {d}_K\bigl(x, \fdc(Y,z), \psi(x,z,u)\bigr)}\end{pmatrix} 
\end{IEEEeqnarray} 
  where the expectation is with respect to $P_{Y|X}(\cdot|x)$. Let $\set{S}^{(x,z)}$ denote the image of $h^{(x,z)}$:  
  \begin{equation}\label{eq:Sxz}
  \set{S}^{(x,z)} \triangleq \bigl\{ \vect{s} \in \Reals^K_+ \colon\vect{s}= h^{(x,z)}(u) \textnormal{ for some  }u\in\set{U}\bigr\}.
  \end{equation}
By definitions~\eqref{eq:dkxz}--\eqref{eq:Sxz}
  \begin{equation}
  \begin{pmatrix} 
  D_1^{(x,z)} \\ \vdots \\ D_K^{(x,z)} 
  \end{pmatrix} \in \textnormal{conv}\bigl(\set{S}^{(x,z)}\bigr)
  \end{equation}
  and,
  consequently, 
  there exists a point
  \begin{equation*}
   \bar{\vect{s}}^{(x,z)}= \begin{pmatrix} 
  \bar{s}_{1}^{(x,z)} \\ \vdots\\  \bar{s}_{K}^{(x,z)}   \end{pmatrix}
  \end{equation*} 
on the boundary of conv$(\set{S}^{(x,z)})$ with
  \begin{equation}\label{eq:barsD}
    \bar{s}_{k}^{(x,z)} \leq D_{k}^{(x,z)}, \qquad k\in\{1,\ldots, K\}.
  \end{equation} 
  Since $\set{S}^{(x,z)}$ is compact (it contains at most
  $|\hat{\set{X}}_\e|$ points because $h^{(x,z)}(u)$ depends on $u$
  only via $\psi(x,z,u)$), Lemma~\ref{lem:card} implies that
  $\bar{\vect{s}}^{(x,z)}$ can be written as a convex combination of
  $K$ or fewer points in $\set{S}^{(x,z)}$:
   \begin{equation}\label{eq:linearlambda}
   \bar{\vect{s}}^{(x,z)} = \sum_{j=1}^K \lambda_j \, \vect{s}_j^{(x,z)},
   \end{equation}
   where $\vect{s}_1^{(x,z)} , \ldots, \vect{s}_K^{(x,z)} \in
   \set{S}^{(x,z)}$ and the coefficients $\lambda_1, \ldots \lambda_K
   \in [0,1]$ sum to 1.  Let $u_1^{(x,z)}, \ldots, u_K^{(x,z)}\in
   \set{U}$ be preimages of $\vect{s}_1^{(x,z)} , \ldots,
   \vect{s}_K^{(x,z)}$ so
  \begin{equation}\label{eq:imagesk}
h^{(x,z)}\big(u_j^{(x,z)}\big)= \vect{s}_j^{(x,z)}, \qquad j\in\{1,\ldots, K\}.
  \end{equation}
  
  We can now define the function $\tilde{\psi}$ as mapping every pair
  $(x,z) \in \set{X} \times \set{Z}$ and every $j\in \{1,\ldots, K\}$
  to 
  \begin{equation}\label{eq:psiti}
    \tilde{\psi}(x, z, j) \triangleq {\psi}\bigl(x, z, u_j^{(x,z)}\bigr).
  \end{equation}
  And we define the random
  variable $\tilde{U}$ to be conditionally independent of $Y$ given
  $(X,Z)$ with the conditional law
  \begin{equation}\label{eq:lawU}
  \Prv{\tilde{U} = j | X=x, Z=z} = \lambda_j^{(x,z)}, \qquad j\in \{1,\ldots, K\}.
  \end{equation}
  The Markov condition~\eqref{eq:MarkovU} thus holds by definition.
  Moreover, \eqref{eq:dkxz}, \eqref{eq:funh}, and
  \eqref{eq:barsD}--\eqref{eq:lawU} combine to prove that
  $\tilde{U}$ and $\tilde{\psi}$ also satisfy the $K$ distortion
  constraints in~\eqref{eq:ineqDU}: denoting the $k$-th
  component of the vector $\vect{s}_j$ by $s_{j,k}$, for $j,
  k\in\{1,\ldots, K\}$, 
  \begin{IEEEeqnarray}{rCl}
\lefteqn{\E{ d_k \bigl(x, \fdc(Y,z),\tilde{\psi}(x,z,\tilde{U})
    \bigr)} } 
\quad \nonumber \\
&=& \sum_{j=1}^K \lambda_j \E{ d_k\bigl(x, \fdc(Y,z),\tilde{\psi}(x,z,j)\bigr)} \\
&=& \sum_{j=1}^K \lambda_j \E{ d_k\bigl(x, \fdc(Y,z),\psi(x,z,u_{j}^{(x,z)})\bigr)} \\
&= &  \sum_{j=1}^K  s_{j,k}^{(x,z)} \\
& = & \bar{s}_{k} \\
& \leq & D_{k}^{(x,z)},\label{eq:inDkxz}
  \end{IEEEeqnarray}
  where the first equality holds by~\eqref{eq:lawU}, the second
  equality by~\eqref{eq:psiti}, the third equality by \eqref{eq:funh}
  and \eqref{eq:imagesk}, the fourth equality by
  \eqref{eq:linearlambda}, and the inequality at the end by
  \eqref{eq:barsD}. Finally, from~\eqref{eq:inDkxz} we conclude that
    \begin{IEEEeqnarray}{rCl}\lefteqn{  
        \E{d_k\bigl(X, \fdc(Y,Z),\tilde{\psi}(X,Z,\tilde{U})\bigr)} } \quad \nonumber \\
    & = & \!
 \sum_{x\in\set{X}, z\in\set{Z}}\!\Prv{X=x, Z=z} 
 \E{ d_k\bigl(x, \fdc(Y,z),\tilde{\psi}(x,z,\tilde{U})\bigr)} \nonumber \\\\ & \leq &\! \sum_{x\in\set{X}, z\in\set{Z}}\! \Prv{X=x, Z=z} D_{k}^{(x,z)}\\
 & \leq & D_k
  \end{IEEEeqnarray}
  where the last inequality follows from the definition of
  $D_k^{(x,z)}$ in~\eqref{eq:dkxz} and the fact that the tuple $(U,Z,
  \phi, \psi)$ satisfies the original distortion constraints
  in~\eqref{eq:distErk}.
\end{proof}


\end{document}